\documentclass[12pt,a4paper]{article}
\usepackage[dvipsnames]{xcolor}
\usepackage[intoc,refpage]{nomentbl}
\makenomenclature
\usepackage[final]{hyperref}
\hypersetup{colorlinks=true,citecolor=black}
\usepackage{graphicx}

\usepackage[ngerman,english]{babel}
\usepackage{paralist,hyperref,ifpdf,url,textcomp}
\usepackage{amssymb,amsmath,amscd}
\usepackage[amsmath,thmmarks,hyperref]{ntheorem}
\usepackage[capitalise]{cleveref}
\crefname{equation}{Equation}{Equations}
\Crefname{equation}{Equation}{Equations}
\crefname{examples}{Examples}{Examples}
\Crefname{examples}{Examples}{Examples}
\crefname{figure}{Figure}{Figures}
\Crefname{figure}{Figure}{Figures}
\usepackage{aliascnt}           
\newcommand{\mynewtheorem}[4][]{
  \ifthenelse{\equal{#1}{}}{	
    \newtheorem{#2}{#3}		
  }{
    \newaliascnt{#2}{#1}	
    \newtheorem{#2}[#2]{#3}	
    \aliascntresetthe{#2}	
  }
  \crefname{#2}{#3}{#4}		
}
\numberwithin{equation}{section}
\newtheorem{theorem}{Theorem}[section]
\mynewtheorem[theorem]{lemma}{Lemma}{Lemmas}
\mynewtheorem[theorem]{proposition}{Proposition}{Propositions}
\mynewtheorem[theorem]{corollary}{Corollary}{Corollaries}
\theorembodyfont{\normalfont}
\mynewtheorem[theorem]{defi}{Definition}{Definitions}
\mynewtheorem[theorem]{conjecture}{Conjecture}{Conjectures}
\mynewtheorem[theorem]{example}{Example}{Examples}
\mynewtheorem[theorem]{notation}{Notation}{Notations}
\mynewtheorem[theorem]{remark}{Remark}{Remarks}
\mynewtheorem[theorem]{remarks}{Remarks}{Remarks}
\theoremsymbol{\ensuremath{\Box}}
\theoremseparator{:}
\newtheorem*{proof}{Proof}
\newcommand{\beq}{\begin{equation}}
\newcommand{\eeq}{\end{equation}}
\newcommand{\Leq}[1]{\label{#1}\end{equation}}
\newcommand{\es}{\emptyset}
\renewcommand {\l}{\left}
\newcommand {\ri}{\right}
\newcommand {\vep}{\varepsilon}

\newcommand {\LA}{\left\langle}
\newcommand {\RA}{\right\rangle}
\newcommand {\pa}{\partial}

\newcommand {\eh}{{\textstyle \frac{1}{2}}}

\newcommand {\ev}{{\textstyle \frac{1}{4}}}
\newcommand {\dv}{{\textstyle \frac{3}{4}}}

\newcommand {\ar}{\rightarrow}
\newcommand {\sign}{\mathrm{sign}}

\newcommand {\ess}{\operatorname{ess}}
\newcommand {\Euclid}{\mathrm{Euclid}}

\newcommand {\Id}{\mathrm{Id}}
\newcommand {\Span}{\mathrm{span}}

\newcommand {\bC}{{\mathbb C}}

\newcommand {\bE}{{\mathbb E}} 
\newcommand {\bN}{{\mathbb N}}
\newcommand {\bR}{{\mathbb R}}
\newcommand {\bZ}{{\mathbb Z}}

\newcommand {\bQ}{{\mathbb Q}}

\newcommand{\rstr}{{\upharpoonright}}
\newcommand{\idty}{{\rm 1\mskip-4mu l}} 
\newcommand{\cB}{{\cal B}} 
\newcommand{\cD}{{\cal D}} %
\newcommand{\cE}{{\cal E}} %
\newcommand{\cH}{{\cal H}}

\newcommand{\cK}{{\cal K}}
\newcommand{\cL}{{\mathcal L}}
\newcommand{\cM}{{\cal M}}
\newcommand{\cN}{{\cal N}}
\newcommand{\cO}{{\cal O}} 

\newcommand{\cS}{{\cal S}} 
\newcommand{\cV}{{\cal V}} %

\newcommand{\cT}{{\cal T}}

\newcommand{\ov}{\overline}
\newcommand{\bem}{\l(\! \begin{array}}
\newcommand{\eem}{\end{array}\!\ri)}
\newcommand{\bsm}{\left(\begin{smallmatrix}} 
\newcommand{\esm}{\end{smallmatrix}\right)}  

\newcommand{\hA}{\hat{A}}
\newcommand{\hS}{\hat{\Sigma}}

\newcommand{\q}{{\vec{q}\,}}

\newcommand{\tmu}{{\hat{\mu}}}
\newcommand{\tlambda}{{\hat{\lambda}}} 
\newcommand{\tP}{{\hat{P}}}
\newcommand{\tp}{{\hat{p}}}
\newcommand{\tpi}{{\hat{\pi}}}
\newcommand{\tB}{{\hat{B}}}
\newcommand{\tiB}{{\tilde{B}}}
\newcommand{\tH}{{\hat{H}}}
\newcommand{\tM}{{\hat{M}}}	
\newcommand{\tV}{{\hat{V}}}	
\newcommand{\tW}{{\hat{W}}}	
\newcommand{\tPhi}{{\hat{\Phi}}}

\newcommand{\ovv}{{\overline{v}}}
\newcommand{\qtext}[1]{\quad\text{#1}\quad}
\newcommand{\qtextq}[1]{\quad\text{#1}\quad}
\RequirePackage{ifthen}
\newcommand{\abs}[2][]{\lvert#2\rvert\ifthenelse{\equal{#1}{}}{}{_{#1}}}
\newcommand{\bigabs}[2][]{\bigl|#2\bigr|\ifthenelse{\equal{#1}{}}{}{_{#1}}}
\newcommand{\Bigabs}[2][]{\Bigl|#2\Bigr|\ifthenelse{\equal{#1}{}}{}{_{#1}}}
\newcommand{\norm}[2][]{\lVert#2\rVert\ifthenelse{\equal{#1}{}}{}{_{#1}}}
\newcommand{\setsize}[2][]{\lvert#2\rvert\ifthenelse{\equal{#1}{}}{}{_{#1}}}

\newcommand{\dist}{\mathrm{dist\,}}
\newcommand{\supp}{\mathrm{supp}}

\providecommand{\Vintern}[1]{\left(\begin{smallmatrix}#1\end{smallmatrix}\right)}%
\providecommand{\V}[1]{{\mathchoice{\begin{pmatrix}#1\end{pmatrix}}{\Vintern{#1}}{\Vintern{#1}}{\Vintern{#1}}}}
\providecommand{\xto}{\xrightarrow}
\providecommand{\Union}{\bigcup}

\providecommand{\Isect}{\bigcap}
\providecommand{\isect}{\cap}

\providecommand{\ve}{\varepsilon}
\providecommand{\Mo}{\mathbf M_\omega}
\providecommand{\textq}[1]{\text{#1}\quad}

\begin{document}
\title {Classical Motion in Random Potentials}
\author{Andreas Knauf
\thanks{Department Mathematik, 
Universit\"at Erlangen-N\"urnberg,
Cauerstr.~11, D--91\,058 Erlangen, Germany.
e-mail: knauf@mi.uni-erlangen.de}
\and Christoph Schumacher
\thanks{Fakult\"at f\"ur Mathematik, 
Technische Universit\"at Chemnitz,
Reichenhainerstr.~$41$, D--09\,126 Chemnitz, Germany.
e-mail: christoph.schumacher@mathematik.tu-chemnitz.de}
}
\date{October 2011}
%
\maketitle
\begin{abstract}
  We consider the motion of a classical particle under the influence
  of a random potential on~$\bR^d$, in particular the distribution of
  asymptotic velocities and the question of ergodicity of time evolution.
\end{abstract}
\tableofcontents
%
\section{Introduction}\label{sec1}
%
Since its introduction to physics by Einstein
and its mathematical foundation by Wiener, Brownian motion is 
considered one cornerstone of probability and of thermodynamics.
It thus may seem natural to expect 
that classical motion in a spatially
homogeneous force field leads in the large time limit
to deterministic diffusion.

Such results can in fact be proven for a periodic Lorentz gas
with finite horizon (related to the Sinai billiard) 
or for coulombic periodic potentials, 
see \cite{BS81} respectively \cite{Kna87, DL91}.

However, in both cases the dynamics is non--smooth, of billiard
type for the Lorentz gas and with orbits locally approximating
Keplerian conic sections for the second case.
We show here (\cref{theo}), that the 
motion in a bounded smooth potential 
is incompatible with uniform hyperbolicity.

If periodic scatterers can lead to diffusive motion,
the more this should be true for random scatterers. 
However, this is certainly not the case for 1D and
also wrong for more degrees of freedom and Poisson potentials
(\cref{sec5,sec5a}).
In this and other cases the Hamiltonian flow is not even ergodic
(\cref{thmPoisson}).

For the random Coulomb case, however, we show in 
\cref{thm:toptrans} that the flow is 
typically topologically transitive for large energies.

Related results on the motion in random 
configurations of convex scatterers have been derived by 
Marco Lenci and collaborators. 
In \cite{Len03}, the planar situation with a finite modification
of periodic scatterers was studied.
In \cite{CLS10} (see also the references of that article) 
recurrence for particles in quenched tubes
with random scatterers has been proven.

The literature on the corresponding quantum problem of 
Schr\"odinger operators with random potentials is much broader. 
See \emph{e.g.}, \cite{LMW03}, and \cite{Ve08} with its extensive
references.

In the article \cite{DRF11} quantum diffusion in a thermal 
medium has been proven.

We describe the structure of the paper.
Random potentials arise in different guises. 
The one studied most extensively is based on a regular lattice~$\cL$
in configuration space. If $J$ indexes the single site potentials,
on assumes a measure on the space $J^\cL$,
for which the $\cL$--action is ergodic. 
This is studied in \cref{sec2}.
One result is an almost deterministic distribution of 
asymptotic velocities (\cref{sec3}).
One problem concerning the Liouville measure
on the energy surface is discussed in \cref{sec4a}.
Depending on the exact exponential decay rate,
the set of singular energy values may or may not be typically dense.
Poissonian random potentials are studied in \cref{sec4,sec5a}.
Relationships between different notions of ergodicity
and some of their dynamical implications are discussed in \cref{sec5}. 
\cref{sec6} concerns deterministic potentials, and the 
geometry compatible with a uniformly hyperbolic structure.
Finally, we treat random coulombic potentials in \cref{sec9,sect10},
showing topological transitivity and ergodicity of the compactified
flow for energy surfaces of large energy.
\bigskip

\noindent
{\bf Acknowledgement:} We thank Boris Gutkin for the useful conversations.
%
\section{The Lattice Case}\label{sec2}
%
We assume the random potential to be based on \emph{short range} single site
potentials $W_j\in C^\eta(\bR^d,\bR)$,
\nomenclature[AW]{$W_j$}{single site potential}{}{}%
\nomenclature[Ad]{$d$}{dimension}{}{}%
indexed by $j\in J$, $\setsize J\in\bN$, $\setsize J\ge2$%
\nomenclature[AJ]{$J$}{index set for single site potentials}{}{}.
Namely $\eta\in\bN\cup\{\infty\}$, $\eta\ge2$%
\nomenclature[Geta]{$\eta$}{order of differentiability}{}{}
and
\beq
\abs{\pa^\alpha W_j(q)}\leq\frac{C_\alpha}{\LA q\RA^{d+\vep}} \qquad
\bigl(q\in\bR^d,\ \alpha\in\bN_0^d,\ \abs\alpha\leq\eta\bigr)
\Leq{as:sisi}
\nomenclature[ACalpha]{$C_\alpha$}{constant}{}{}%
\nomenclature[Galpha]{$\alpha$}{multi index}{}{}%
with $\LA q\RA:=\sqrt{1+\norm{q}^2}$, for constants $C_\alpha>0$.

The single site potentials are placed on a regular lattice $\cL\subseteq\bR^d$%
\nomenclature[AL]{$\cL$}{regular lattice in $\bR^d$}{}{}
with basis $\ell_1,\dotsc,\ell_d\in\bR^d$%
\nomenclature[Aell1]{$\ell_1,\dotsc,\ell_d$}{basis of $\cL$}{}{}
according to $\omega\in\Omega:=J^\cL$%
\nomenclature[GOmega]{$\Omega$}{configurations of single site potentials}{}{}%
\nomenclature[Gomega]{$\omega$}{configuration of single site potentials}{}{}
to give the random potential
\[V\colon M:=\Omega\times\bR^d\to\bR \qtextq{,} V(\omega,q):=\sum_{\ell\in\cL}
W_{\omega(\ell)}(q-\ell)\]
\nomenclature[AV]{$V$}{(random) potential}{}{}%
\nomenclature[Al]{$\ell$}{lattice vector}{}{}%
on \emph{extended configuration space}~$M$%
\nomenclature[AM]{$M$}{extended configuration space}{}{}.
We use the discrete topology on~$J$ and product topology on~$\Omega$.
The probability measure~$\beta$%
\nomenclature[Gbeta]{$\beta$}{probability measure on $\Omega$}{}{}
on $\bigl(\Omega,\cB(\Omega)\bigr)$
is assumed to be invariant w.r.t.\ the action
\begin{equation}\label{eq:vartheta}
  \vartheta\colon\cL\times\Omega\to\Omega\qtextq{,}
   (\ell,\omega)\mapsto\vartheta_\ell(\omega)\qtextq{with}
  \vartheta_\ell(\omega)(\ell'):=\omega(\ell'+\ell)\text.
\end{equation}
\nomenclature[Gtheta]{$\vartheta$}{$\cL$--action on $\Omega$}{}{}%
Unless we explicitly say the contrary, $\beta$ is assumed to be 
$\vartheta$--ergodic (a simple example being a product measure
$\beta=\bigotimes_{\ell\in\cL}\hat{\beta}$
with a probability measure~$\hat{\beta}$%
\nomenclature[Gbetahat]{$\hat\beta$}{probability measure on~$J$}{}{}
on~$J$).
Ergodicity and $\setsize J\ge2$ imply that~$\beta$ is non--atomic.
The short range conditions above imply that for all $\omega\in\Omega$
the potentials
\[V_\omega\colon\bR^d\to\bR\qtextq{,} V_\omega(q):=V(\omega,q)\]
are as smooth as the single site potentials $W_j$.
Moreover~$V$ itself is continuous and bounded,
together with its partial derivatives $\pa^\alpha V,\ \abs\alpha\leq\eta$.
\begin{remark}\label{rem:A}
More relevant than the \emph{range} 
\begin{equation*}
  [V_{\min},V_{\max}]:=\ov{V(M)}
\end{equation*}
\nomenclature[AVmin]{$V_{\min}$}{infimum of~$V$}{}{}%
\nomenclature[AVmax]{$V_{\max}$}{supremum of~$V$}{}{}%
of the potential is its \emph{essential range}
($\lambda^d$%
\nomenclature[Glambdad]{$\lambda^d$}{Lebesgue measure on $\bR^d$}{}{}
denoting Lebesgue measure on $\bR^d$):
\begin{equation*}
  \supp\bigl(V(\beta\otimes\lambda^d)\bigr)
    =[V_{\ess\min},V_{\ess\max}]\text.
\end{equation*}
\nomenclature[AVessmin]{$V_{\ess\min}$}{essential infimum of~$V$}{}{}%
\nomenclature[AVessmax]{$V_{\ess\max}$}{essential supremum of~$V$}{}{}%
By $\vartheta$--ergodicity of $\beta$, this is a deterministic set:
\[\supp \bigl(V_\omega(\lambda^d)\bigr)
 = \supp\bigl(V(\beta\otimes\lambda^d)\bigr) \qquad 
 (\beta\text{--a.s.})\text.\]
As~$V$ is uniformly bounded together with its first derivatives, this is in fact
a bounded interval: 
$\supp\bigl(V_\omega(\lambda^d)\bigr)=\ov{V_\omega(\bR^d)}$.
In general the essential range of~$V$ is a proper subinterval of its range.
\end{remark}
Thus the flow $\Phi\colon\bR\times P\to P$%
\nomenclature[GPhi]{$\Phi$}{hamiltonian flow}{}{}
generated by the Hamiltonian function
\[H\colon P\to\bR \qtextq{,} H(\omega,p,q):=\eh\norm p^2+V(\omega,q)\]
\nomenclature[AH]{$H$}{hamiltonian function}{}{}%
on \emph{extended phase space}
\[P:=\Omega\times\bR^d\times\bR^d\]
\nomenclature[AP]{$P$}{extended phase space}{}{}%
uniquely exists for all times.
We write 
\[\Phi^t\colon P\to P\text{ , } \Phi^t(\omega,p_0,q_0):=\Phi(t,\omega,p_0,q_0)= 
\l(\omega\, ,\, p^t(\omega,p_0,q_0)\, ,\, q^t(\omega,p_0,q_0)\ri)\]
\nomenclature[Apt]{$p^t$}{momentum at time~$t$}{}{}%
\nomenclature[Aqt]{$q^t$}{location at time~$t$}{}{}%
for the solution of the initial value problem at time $t\in\bR$%
\nomenclature[At]{$t$}{time}{}{}.
Whenever a fixed $\omega\in\Omega$ is considered, we write it as a subscript
(\emph{e.g.}\ $\Phi_\omega^t\equiv(p_\omega^t,q_\omega^t)\colon\bR^{2d}\to\bR^{2d})$.
See 
\cref{abb:realisation}
for a realisation of $t\mapsto q_\omega(t,x_0)$, with lattice $\cL=\bZ^2$.
\begin{figure}[h]
\begin{center}
  \ifpdf
    \includegraphics[viewport=90 375 440 715,width=9cm,clip]{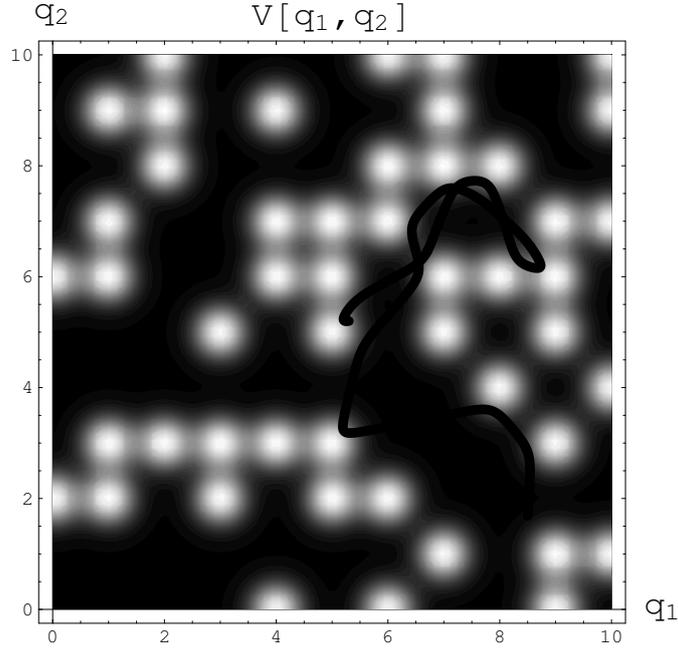}
  \else
    \includegraphics[viewport=90 375 440 715,width=9cm,clip]{zufall4.ps}
  \fi
\end{center}
  \caption{Path in configuration space}
  \label{abb:realisation}
\end{figure}

The space~$P$ is equipped with the locally finite Borel measure
$\mu:=\beta\otimes\lambda^{2d}$%
\nomenclature[Gmu]{$\mu$}{invariant measure on~$P$}{}{}.
The Hamiltonian flow~$\Phi$ leaves~$\mu$ invariant.

The lattice acts on extended phase space via the continuous group action
\[\Theta\colon\cL\times P\to P \qtextq{,}(\ell,\omega,p,q)\,\mapsto\,
\Theta_\ell(\omega,p,q) := \bigl(\vartheta_\ell(\omega),p,q-\ell\bigr)\]
\nomenclature[GTheta]{$\Theta$}{$\cL$--action on~$P$}{}{}%
and leaves $\mu$ invariant:
\beq
\mu\circ\Theta_\ell=\mu \qquad (\ell\in\cL)\text.
\Leq{A}
Similarly for all lattice vectors $\ell\in\cL$
\beq
H\circ\Theta_\ell=H \qtextq{and thus} \Theta_\ell\circ\Phi^t=\Phi^t\circ
\Theta_\ell \qquad (t\in\bR)\text.
\Leq{B}

By the Hamiltonian character of the motion the energy is invariant under 
the time evolution:
\beq
H\circ\Phi^t=H \qquad (t\in\bR)\text.
\Leq{C}
Thus we obtain a one--parameter family of Borel measures $\mu_E$%
\nomenclature[GmuE]{$\mu_E$}{invariant measure on $H^{-1}(-\infty,E]$}{}{}
on~$P$, given by
\beq
\mu_E(B):=\mu \l(B\cap H^{-1}(-\infty,E])\ri) 
 \qquad \bigl(B\in\cB(P)\bigr)\text,
\Leq{mu:E}
parametrised by the energy $E\in\bR$%
\nomenclature[AE]{$E$}{energy}{}{}.
By \eqref{A}, \eqref{B} and \eqref{C}
these measures~$\mu_E$ are $\Theta$--~and $\Phi$--invariant, too.

A relevant quantity is \emph{asymptotic} velocity
\beq
\ovv^\pm\colon P\to\bR^d \qtextq{,} \ovv^\pm(\omega,x_0):=\lim_{T\to\pm\infty}
\frac{q_\omega(T,x_0)}{T}\text.
\Leq{D}
\nomenclature[Avbar]{$\ovv$}{asymptotic velocity}{}{}%

\begin{remark}[Non--Existence of Asymptotic Velocity] \label{rem:never}
For Hamiltonian motion in bounded potentials the asymptotic velocity may not exist for \emph{any} initial condition 
$x_0\in H^{-1}(E)$ for \emph{all} energies $E$ above some threshold. 
This is the case for centrally symmetric potentials
($V(q)=\tilde V(\norm q)$), with $\tilde V$%
\nomenclature[AVtilde]{$\tilde V$}{slowly varying function}{}{}
being periodic in a slowly varying function like, for example, 
$\tilde V(r)=\cos(\log(r+1))$.

It is also possible to construct examples of truly random potentials 
where for some $\omega\in\Omega$ these limits do not exist on 
$ H_\omega^{-1}(E)$ for any $E>E_0$:

The lattice $\cL:=\bZ^d$ with fundamental domain $\cD=[0,1]^d$%
\nomenclature[AD]{$\cD$}{fundamental domain of~$\cL$}{}{}
admits an adapted partition of unity $F_\ell$%
\nomenclature[AFell]{$F_\ell$}{partition of unity}{}{},
$(\ell\in\cL)$
\[{\textstyle \sum_{\ell\in\cL}} F_\ell(q) = 1 \qquad (q\in\bR^d)\]
with $F_\ell(q) := F(q-\ell)$
for 
$F:=\idty_\cD\,*\,f$ with some
$f\in C_c^\infty\bigl(\bR^d,[0,\infty)\bigr)$ and $\int_{\bR^d}f\,dx=1$.
\nomenclature[Aidentity]{$\idty_\cD$}{indicator function of~$\cD$}{}{}%
\nomenclature[Af]{$f$}{smooth nonnegative function}{}{}%

We now take $J:=\{0,1,\ldots,d\}$ as index set for the single site potentials
$W_j:=j\,F$.
We choose $\omega\in\Omega=J^{(\bZ^d)}$ of the form 
$\omega(\ell)=\sum_{k=1}^{d} \widetilde\omega_k(\ell_k)$%
\nomenclature[Gomegatildek]{$\widetilde\omega_k$}{slowly varying configuration}{}{},
with $\widetilde\omega_k:\bZ\to \{0,1\}$ slowly varying like above.

Then the motion in the realization $V_\omega$
is separable in cartesian coordinates,
and thus for no $E>E_0:=d$%
\nomenclature[AEzero]{$E_0$}{energy threshold}{}{}
and initial condition $x_0\in H_\omega^{-1}(E)$
the asymptotic velocities $\ovv^\pm(x_0)$ exist.

By allowing for arbitrary modifications of $\omega$ over finite subsets
of~$\bZ^d$ the set of these $\widehat{\omega}$%
\nomenclature[Gomegahat]{$\widehat\omega$}{finite modification of $\omega$}{}{}
is even dense in $\Omega$,
and asymptotic velocity does not exist for any initial condition
$x_0\in H_{\widehat{\omega}}^{-1}(E)$
(possibly except for dimension $d>1$ and a set of finite Liouville measure,
for which $\ovv^\pm(\widehat{\omega},x_0)=0$).
\end{remark}
In case of non--existence of the limit, we will set $\ovv^\pm(\omega,x_0):=0$.

However, typically the above limit exists.
This is shown below by invoking Birkhoff's ergodic theorem.
As this deals with finite measures, we have to change our measure space.
Namely, we set 
\[\hat{P}:=P/\Theta\text.\] 
\nomenclature[APhat]{$\hat P$}{quotient of $P$ by $\cL$}{}{}%
As $P=\bigcup_{\ell\in\cL}\Theta_\ell(\Omega\times\bR^d_p\times\cD)$ 
with \emph{fundamental domain} of~$\cL$
\[\cD\,:=\,\bigl\{{\textstyle \sum_{k=1}^d} x_k\ell_k\ \bigm| \ x_k\in[0,1]\bigr\}
\,\subseteq\,\bR^d_q\text,\] 
the topological space~$\hat{P}$ is homeomorphic to
\beq
(\Omega\times\bR^d\times \cD)/{\sim}\text,
\Leq{homeo}
$\sim$ identifying points in $\Omega\times\bR^d_p\times\pa\cD$
via the $\Theta$ action.
Using \eqref{A} and \eqref{B}, the covering projection 
\[\hat{\pi}\colon P\to\hat{P}\]
\nomenclature[Gpihat]{$\hat\pi$}{covering projection $P\to\hat P$}{}{}%
allows us to induce measures $\tmu$%
\nomenclature[Gmuhat]{$\tmu$}{invariant measure on $\hat P$}{}{}
and $\tmu_E$%
\nomenclature[GmuhatE]{$\tmu_E$}{invariant measure on $\hat P$}{}{}
on $\hat{P}$ by setting
\beq
  \tmu(\tB):=\mu\bigl(\hat{\pi}^{-1}(\tB)\cap(\Omega\times\bR^d_p\times\cD)\bigr)
  \qquad\bigl(\tB\in\cB(\tP)\bigr)
\Leq{mu:hat}
and similarly for $\tmu_E$ (note that the Lebesgue measure 
$\lambda^d(\pa\cD)=0$).

Furthermore \eqref{B} allows us to define the continuous energy function
$\tH\colon\tP\to\bR$%
\nomenclature[AHhat]{$\tH$}{hamiltion function on $\hat P$}{}{}
and flow $\tPhi\colon\bR\times\tP\to\tP$%
\nomenclature[GPhitilde]{$\tPhi$}{hamiltion flow on $\hat P$}{}{}
uniquely by
$H=\tH\circ\hat{\pi}$ and $\tPhi^t\circ\hat{\pi}=\hat{\pi}\circ\Phi^t$.
Finally, by $\Theta$--invariance of the momenta $p\colon P\to\bR^d$,
they descend to
\beq
\tp\colon\tP\to\bR^d \qtextq{,} \tp\circ\hat{\pi}=p.
\Leq{hat:p}
\nomenclature[Aphat]{$\tp$}{momentum on $\hat P$}{}{}%
\begin{proposition}\label{prop:A}
The asymptotic velocities \eqref{D} exist $\mu$--a.e.\ on~$P$, and
$\ovv^+=\ovv^-\ \mu$--a.e.\,.
Setting $\ovv(x):=\ovv^\pm(x)$ in case of equality and $\ovv(x)=0$ otherwise, 
\[\ovv\in L_{\mathrm{loc}}^\infty(P,\mu)\qtextq{,}
\ovv\circ\Phi^t=\ovv\qtextq{and} \ovv\circ\Theta_\ell=\ovv
\qquad(t\in\bR,\ell\in\cL)\text.\]
\end{proposition}
\begin{proof}
  For initial conditions $x_0=(p_0,q_0)\in \bR^d_p\times \bR^d_q$ we have
  \begin{align*}
    \ovv^\pm(\omega,x_0)&
      =\lim_{T\to\pm\infty}\frac{q_\omega(T,x_0)-q_0}{T}
      =\lim_{T\to\pm\infty}\frac{1}{T}\int_0^Tp_\omega(t,x_0)\,dt\\&
      =\lim_{T\to\pm\infty}\frac{1}{T}\int_0^T\tp\bigl(t,\hat{\pi}(\omega,x_0)\bigr)\,dt
      \text.
  \end{align*}
  But as the finite measures $\tmu_E$ on $\tP$ are $\tPhi$--invariant,
  by Birkhoff's ergodic theorem the limits $\lim_{T\to\pm\infty}\frac{1}{T}
  \int_0^T\tp(t,\hat x_0)\,dt$ exist, coincide and are $\tPhi$--invariant 
  $\tmu_E$--a.e.\,.
  Moreover the limit function lies in $L^1(\tP,\tmu_E)$ and is bounded by
  $\sqrt{2(E-V_{\min})}$ in absolute value, as 
  $\norm{\tp(t,\hat x_0)}^2\leq2 (\tH(\hat x_0)-V_{\min})$.
  Thus it is in $L^\infty(\tP,\tmu_E)$.

  The measure $\tmu$ is non--finite, but 
  $\tmu\rstr_{\tH^{-1}((-\infty,E])} =\tmu_E$.
  Thus we obtain the result.
\end{proof}
\begin{corollary}\label{cor:A}
For $\beta$--a.e.\ $\omega\in\Omega$ the asymptotic velocities
$\ovv_\omega^\pm\colon \bR^{d}_p\times\bR^{d}_q\to\bR^d$ exist and are equal
$\lambda^{2d}$--a.e.\,.
\end{corollary}
\begin{proof}
  This follows from \cref{prop:A} and Fubini's theorem, applied to
  the measure $\mu=\beta\otimes\lambda^{2d}$ on extended phase space~$P$.
\end{proof}
Similarly, for all regular energies $E$ (see \cref{sec5} below), 
the asymptotic velocities exist Liouville--almost everywhere on $H^{-1}_\omega(E)$.
%
\section{Distribution of Asymptotic Velocities}\label{sec3}
%
Next we consider the joint distribution of energy and asymptotic velocity,
using the measurable maps
\[\Gamma_\omega:=(H_\omega,\ovv_\omega)\colon\bR^{d}_p\times\bR^{d}_q\longrightarrow\bR\times\bR^d \qquad
(\omega\in\Omega)\text.\]
\nomenclature[GGammaomega]{$\Gamma_\omega$}{energy velocity map}{}{}%
We thus consider phase space regions 
$\bR^d_p\times Q_n\subseteq\bR^{d}_p\times\bR^{d}_q$ with
\begin{equation*}
  Q_n:=\bigl\{{\textstyle \sum_{k=1}^d} x_k\ell_k\ \bigm|\ 
    \forall k\in\{1,\dotsc,d\}:\; -n\leq x_k<n\bigr\}
    \qquad(n\in\bN)\text,
\end{equation*}
\nomenclature[AQn]{$Q_n$}{rectangle in $\bR_q^d$}{}{}%
and the normalised restrictions
\[\lambda_n\, :=\,
 \frac{1}{\lambda^d(Q_n)}\lambda^{2d}\rstr_{\bR^d_p\times Q_n}
 \,=\, (2n)^{-2d}\; \lambda^{2d}\rstr_{\bR^d_p\times Q_n}\]
\nomenclature[Glambdan]{$\lambda_n$}{normalized Lebesgue measure on~$Q_n$}{}{}%
of Lebesgue measure to these regions.
\begin{proposition}\label{prop:B}
For $\beta$--a.e.\ $\omega\in\Omega$ the {\bf energy velocity distribution}
\[\nu_\omega:=\lim_{n\to\infty}\Gamma_\omega(\lambda_n)\]
\nomenclature[Gnuomega]{$\nu_\omega$}{energy velocity distribution}{}{}%
exists in the sense of vague convergence and is independent of $\omega$.
\end{proposition}
\begin{proof}
  In view of Riesz representation theorem
  we have to check that for every function
  $f\in C_c(\bR^{d+1})=\{g\in C(\bR^{d+1})\mid\supp(g)\text{ is compact}\}$
  \[\lim_{n\to\infty}\int_{\bR^{d+1}}f\,d\nu_{n,\omega}\qquad \text{, with $\nu_{n,\omega}:=\Gamma_\omega(\lambda_n)$,}
  \]
  exists and is independent of~$\omega$.
  By compactness of support of~$f$, there exists an $E_0\in\bR$ (depending
  on~$f$) with $f(E,\ovv)=0$ for all $\ovv\in\bR^d$ and $E\ge E_0$.
  On the other hand we know that 
  $\norm{p}\leq\sqrt{2(E_0-V_{\min})}$ if $(p,q)\in H^{-1}([V_{\min},E_0])$.
  Thus we have the estimate 
  \[\int_{\bR^{d+1}}\abs f\,d\nu_{n,\omega}\leq\frac{(2(E_0-V_{\min})\,\pi)
  ^{d/2}}{\Gamma\l(\frac{d}{2}+1\ri)}\cdot\norm[\infty]f\text,\]
\nomenclature[GGamma]{$\Gamma$}{Gamma function}{}{}%
  which is uniform in $n\in\bN$ and $\omega\in\Omega$.
  The function
  \[g\colon\Omega\to\bR \qtextq{,} g(\omega):=\int_{\bR^{d+1}}f\,d\nu_{1,\omega}\]
  is thus in $L^\infty(\Omega,\beta)$.
  Moreover
  \[\int_{\bR^{d+1}}f\,d\nu_{n,\omega}=\frac{1}{\setsize{\cL_n}}\sum_{\ell\in\cL_n}
  g\bigl(\vartheta_\ell(\omega)\bigr)\]
  with $\cL_n:=\bigl\{\sum_{k=1}^dn_k\ell_k\mid n_k\in\{-n,-n+1,\ldots,n-1\}\bigr\}
  \subseteq\cL$%
\nomenclature[ALn]{$\cL_n$}{rectangle in $\cL$}{}{}.

  By assumption the probability measure~$\beta$ is ergodic w.r.t.\ the
  $\vartheta$--action on~$\Omega$.
  Thus a lattice version of Birkhoff's
  ergodic theorem (see, \emph{e.g.}\ Keller \cite{Kel98}, Thm.\ 2.1.5) assures that 
  \[\lim_{n\to\infty}\int_{\bR^{d+1}}f\,d\nu_{n,\omega}\ \text{ exists\ and\ 
  is\ $\omega$--independent}\]
  $\beta$--almost surely.
\end{proof}
\vspace{2mm}
The energy--velocity distribution is the unique non--random limit measure
\[\hat{\nu}\colon\cB(\bR^{d+1})\to[0,\infty]\text.\]
\nomenclature[Gnuhat]{$\hat\nu$}{non--random energy velocity distribution}{}{}%
This measure always has a certain symmetry property:
\begin{proposition}
  The energy--velocity distribution $\hat{\nu}$ is invariant w.r.t.\ the 
  inversion of velocity
  \[I\colon\bR^{d+1}\to\bR^{d+1} \qtextq{,} (E,\ovv)\mapsto(E,-\ovv)\text.\]
\end{proposition}
\begin{proof}
  Consider $\omega\in\Omega$ with $\nu_\omega=\hat{\nu}$
  and $\ovv_\omega^+(p_0,q_0)=\ovv_\omega^-(p_0,q_0)$
  for $\lambda^{2d}$--a.e.\ phase space point $(p_0,q_0)\in\bR^{2d}$.
  (By \cref{prop:B,cor:A} $\beta$--a.e.~$\omega$ meets these conditions.)

  By \emph{reversibility} of the flow
  $\Phi_\omega^t=(p_\omega^t,q_\omega^t)$, that is
  \[p_\omega^t(-p_0,q_0)=-p_\omega^{-t}(p_0,q_0) \qtextq{,}
  q_\omega^t(-p_0,q_0)=q^{-t}_\omega(p_0,q_0)\,\text,\]
  we have $\ovv_\omega^+(-p_0,q_0)=-\ovv_\omega^-(p_0,q_0)$.
  Together this gives
  \[\ovv_\omega(-p_0,q_0)=-\ovv_\omega(p_0,q_0)
  \qquad (\lambda^{2d}\text{--a.e.})\text.\]
  On the other hand the phase space region $\bR^d\times Q_n$ as well as the 
  measure $\lambda_n$ on it are invariant w.r.t.\ the antisymplectic 
  transformation $(p,q)\mapsto(-p,q)$ on phase space.
  Thus the image measures $\nu_{n,\omega}= \Gamma_\omega(\lambda_n)$ 
  are $I$--invariant.
  This carries over to the vague limit $\nu_\omega=\hat{\nu}$.
\end{proof}
\begin{defi}
For $\omega\in\Omega$ a phase space point $x_0=(p_0,q_0)\in\bR^{2d}$ is called
\emph{forward} resp.\ \emph{backward bounded} if
\[q_\omega\bigl([0,\infty),x_0\bigr) \qtextq{resp.} 
  q_\omega\bigl((-\infty,0],x_0\bigr)\]
are bounded subsets of configuration space $\bR^d$.
\end{defi}
Note that $\lambda^{2d}$--a.e.\ $x_0\in\bR^{2d}$ is simultaneously bounded
or unbounded in both time directions.
This is a direct consequence of the flow invariance of $\lambda^{2d}$.
\begin{proposition}
  For $E>V_{\ess\max}$ and for $\beta$--a.e.\ $\omega\in\Omega$ for
  every initial position $q_0\in\bR^d$ there exists an initial direction
  $p_0\in\bR^d$ with $H_\omega(p_0,q_0)=E$ and $(p_0,q_0)$ forward 
  unbounded, with positive minimal speed (\/$\inf_{t>0}\norm{(q(t)-q_0)/t}>0$).
\end{proposition}
\begin{proof}
  For $\beta$-a.a.\ $\omega\in\Omega$ we have
  $V_\omega(q)\leq V_{\ess\max}$ for all $q\in\bR^d$
  (see \cref{rem:A} and use the continuity of $V_\omega$).
  Thus the Jacobi metric $g_{E,\omega}$%
  \nomenclature[AgEomega]{$g_{E,\omega}$}{Jacobi metric}{}{}
  on configuration space~$\bR^d$, given by
  \begin{equation}\label{eq:JacobiMetric}
    g_{E,\omega}(q)=\bigl(E-V_\omega(q)\bigr)g_\Euclid\text,
  \end{equation}
  is non--degenerate,
  and the riemannian manifold $(\bR^d,g_{E,\omega})$ is geodesically complete, 
  since the conformal factor is bounded below by $E-V_{\ess\max}>0$.
  So the Hopf--Rinow theorem (see, \emph{e.g.}, \cite{GHL48}, Thm.\ 2.103) implies
  for any $q_n\in\bR^d$ the existence of an initial direction 
  $v_n\in S_{q_0}:=\{v\in\bR^d\mid g_{E,\omega}(q_0)(v,v)=1\}$%
\nomenclature[ASqzero]{$S_{q_0}$}{$g_{E\omega}$--sphere at $q_0$}{}{}
  such that the forward geodesic $t\mapsto\gamma(t,q_0,v_n)$%
\nomenclature[Ggamma]{$\gamma$}{geodesic}{}{}
  with initial condition $(q_0,v_n)$ meets~$q_n$ first at a time~$t_n\ge0$,
  and is a shortest such geodesic so that 
  \[\norm{\gamma(t,q_0,v_n)-q_0}\ge \sqrt{2(E-V_{\ess\max})}\,t 
  \qquad \bigl(n\in\bN,\ t\in[0,t_n]\bigr)\text.\] 
  For $\lim_{n\ar\infty} \norm{q_n-q_0}=\infty$,
  $\lim_{n\ar\infty} t_n=+\infty$, and
  by compactness of $S_{q_0}$ there is an accumulation point
  $v_\infty\in S_{q_0}$ of the $v_n$
  leading to forward unbounded geodesic motion.
  \par
  Up to parametrisation, the geodesic of
  $g_{E,\omega}$ with initial condition $(q_0,v_\infty)$
  coincides with the trajectory $q_\omega(t,x_0)$ (with
  $x_0 := \bigl(2\bigl(E-V_\omega(q_0)\bigr)\,v_\infty,q_0\bigr)$ 
  as initial condition), see \cite[Thm.\ 3.7.7]{AM78}.
  The positivity of the minimal speed is preserved.
\end{proof}
In dimensions $d\ge2$ bounded and unbounded motion can coexist 
$\beta$--a.s.\ for energies $E>V_{\ess\max}$ as well as for $E<V_{\ess\max}$.
In one dimension this is not possible:
\begin{proposition}\label{prop:C}
  For $d=1$ and $\beta$--a.e.\ $\omega\in\Omega$ the motion 
  through $x_0=(p_0,q_0)$ is 
  bounded if $E:=H_\omega(x_0)<V_{\ess\max}$ and 
  unbounded with asymptotic velocity
  \begin{equation*}\label{E}
    \ovv_\omega(x_0)
      =\frac{\ell_1}{\bE_\beta\bigl(\tau(x_0)\bigr)}\qtextq{,}
    \tau_\omega(x_0)
      =\int_0^{\ell_1}\frac{\sign(p_0)}{\sqrt{2(E-V_\omega(q))}}\,dq
  \end{equation*}
\nomenclature[Gtauomega]{$\tau_\omega$}{passage time for fundamental interval}{}{}%
  if $E>V_{\ess\max}$.
\end{proposition}
\begin{remark}
  Here the asymptotic velocity depends on $(\omega,x_0)\in P$
  only via $H_\omega(x_0)$ and $\sign(p_0)$.
  \Cref{abb:periodVel2} 
  shows the shape of $\supp(\hat\nu)\subseteq\bR^2$.
\end{remark}
\begin{figure}[h]
\begin{center}
  \ifpdf
    \includegraphics[viewport=90 466 330 729,width=3cm,clip]{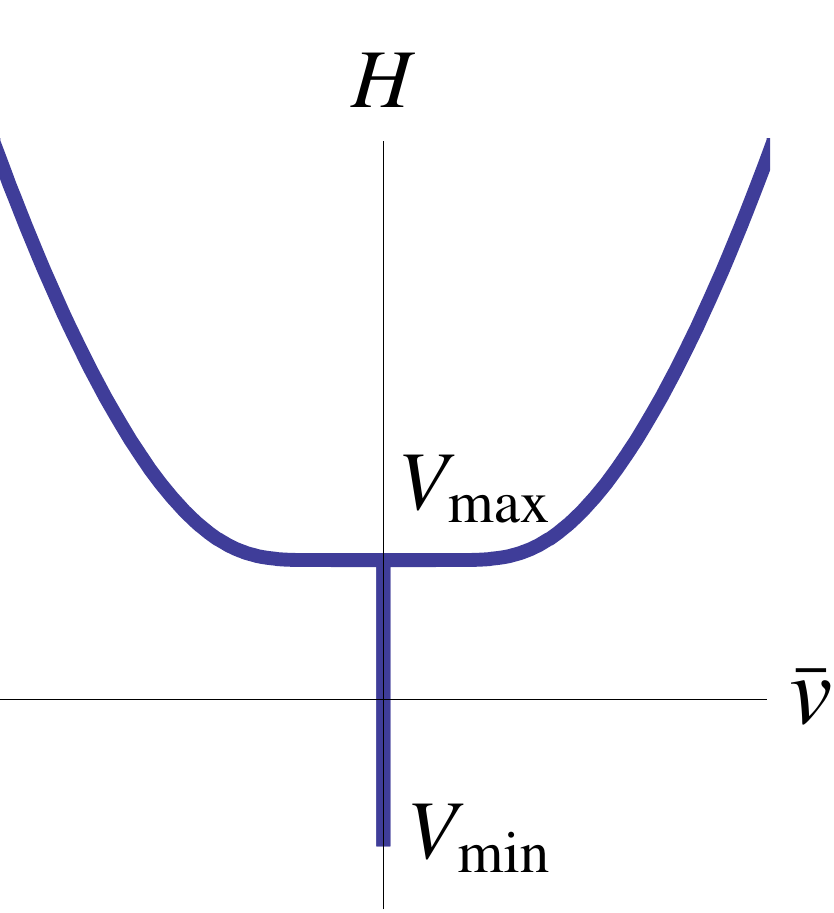}
  \else
    \includegraphics[viewport=90 466 330 729,width=3cm,clip]{periodVel2.ps}
  \fi
\end{center}
  \caption{Shape of $\supp(\hat\nu)\subseteq\bR^2$}
  \label{abb:periodVel2}
\end{figure}
\begin{proof}
\begin{itemize}[$\bullet$]
\item For $E<V_{\ess\max}$ the connected component of $x_0$ in
  $H_\omega^{-1}(E)$ is compact $\beta$--a.s.,
  since there exist $q^\pm\in\bR$ with $V_\omega(q^\pm)>E$ and $q^-<q_0<q^+$.
\item For $E>V_{\ess\max}$ the asymptotic velocity,
  if it exists for $x_0$, equals
  \[\ovv(\omega,x_0)=\lim_{T\to\infty} \frac{1}{T} \int_0^T p_\omega(t,x_0)\,dt
  =\lim_{n\to\infty}\frac{\int_0^{T_n(\omega)}p_\omega(t,x_0)\,dt}
  {T_n(\omega)}\]
  with $T_n(\omega)$ for $n\in\bN$ 
  determined uniquely by $q_\omega(T_n(\omega),x_0)-q_0 =n\ell_1$.
  Here we assume w.l.o.g.\ $\ell_1>0$ and $p_0>0$.
  The numerator equals
  \[\int_0^{T_n(\omega)}\dot{q}_\omega(t,x_0)\,dt = q_\omega(T_n(\omega),x_0)-q_0
  =n\ell_1\text,\]
  whereas
  \[T_n(\omega)=\sum_{k=0}^{n-1}\int_{q_0+k\ell_1}^{q_0+(k+1)\ell_1}\frac
  {1}{\sqrt{2(E-V_\omega(q))}}\,dq\]
  for the denominator.
  Setting 
  \[g\colon\Omega\to\bR \qtextq{,} 
  g(\omega) := \int_{q_0}^{q_0+\ell_1} \frac{1} {\sqrt{2(E-V_\omega(q))}}\,dq\text,\]
  we get $T_n(\omega)=\sum_{k=0}^{n-1}$ $g\bigl(\vartheta_{k\ell_1}(\omega)\bigr)$.
  By continuity of~$g$, we can
  apply Birkhoff's theorem and get by the ergodicity assumption on~$\beta$
  \[\lim_{n\to\infty}\frac{1}{n}T_n(\omega)=\lim_{n\to\infty}\frac{1}{n}
  \sum_{k=0}^{n-1}g\bigl(\vartheta_{k\ell_1}(\omega)\bigr)=\bE(g)
  \qquad \text{$\beta$--a.s.}\text.
  \]
  This proves the assertion for $\lambda^2$--a.e.\ $x_0\in \bR^2$ 
  with $H_\omega(x_0) > V_{\ess\max}$.
  The expression $\frac{1}{T}\int_0^Tp_\omega(t,p_0,q_0)\,dt$
  is monotonically increasing in $p_0$, and our formula for
  $\ovv_\omega(x_0)$
  is continuous in~$p_0$.
  Thus it must be valid for \emph{all}~$x_0$.
\end{itemize}
\end{proof}
For $d\ge2$ it is an interesting question whether for large energies the
asymptotic velocity distribution given by~$\hat\nu$ is zero.
As the example below shows,
this is not always the case for non--trivial random potentials.
\begin{example}[Random Potential With Non-Zero Asymptotic Velocity]\label{ex:la}
We use the cutoff function 
$F\in C_c^\infty(\bR^d,[0,1])$ from \cref{rem:never}
for the lattice $\cL:=\bZ^d$.
Given single site potentials $\tW_j$%
\nomenclature[AWhatj]{$\tW_j$}{spatially restricted single site potential}{}{}
with $\supp(\tW_j)\subseteq\l[\frac{1}{4},\frac{3}{4}\ri]^d$, we set
$W_j(q):=\tW_j(q)-\sum_{k=1}^d\cos(2\pi q_k)F(q)$.
Then the random potential equals
\[V_\omega(q)\; = \; {\textstyle \sum_{\ell\in\cL}}W_{\omega(\ell)}(q-\ell)
\; = \; \tV_\omega(q)-{\textstyle \sum_{k=1}^d}\cos(2\pi q_k)\]
with $\tV_\omega(q):=\sum_{\ell\in\cL} \tW_{\omega(\ell)}(q-\ell)$%
\nomenclature[AVomegahat]{$\tV_\omega$}{spatially restricted random potential}{}{}.

Thus the Hamiltonian function $H_\omega\colon\bR^d\times\bR^d\to\bR$
takes the form
\[H_\omega(p,q) = \tV_\omega(q)+ {\textstyle \sum_{k=1}^d} H_\omega^{(k)}(p_k,q_k)\]
with $H_\omega^{(k)}\colon\bR^d\to\bR,\ (p_k,q_k)\mapsto\frac{1}{2}p_k^2-\cos
(2\pi q_k)$%
\nomenclature[AHomegak]{$H_\omega^{(k)}$}{separated hamiltonian function}{}{}.
The phase space regions
\[P_k:=\bigl\{(p,q)\in\bR^d\times(\bR^d\setminus S)\mid
\forall\,m\in\{1,\ldots,d\}\setminus\{k\}:H_\omega^{(m)}(p_m,q_m)<0\bigl\}\]
\nomenclature[APk]{$P_k$}{phase space region}{}{}%
with $S:=
\cL+\l[\ev,\dv\ri]^d$ %
\nomenclature[AS]{$S$}{enlarged lattice}{}{}%
are invariant w.r.t.\ $\Phi_\omega^t$, since $\tV_\omega(q)=0$ for 
$q\in\bR^d\setminus S$, and thus the motion separates.

For all $x=(p,q)\in P_k$ the asymptotic velocity exists.
If in addition $E:=H_\omega^{(k)}(p_k,q_k)>1$,
then $\ovv^\pm(x)=(0,\ldots,0,\ovv^{(k)},0,\dotsc,0)$ with
$\ovv^{(k)}:=\frac{\sign(p_k)\pi\sqrt{E-1}}{\sqrt{2}\cK(2/(1-E))}$, $\cK$%
\nomenclature[AK]{$\cK$}{complete elliptic integral of the first kind}{}{}
being the complete elliptic integral of the first kind (see \cite{AK98}).

As $E\nearrow\infty$, the intersections $P_k\cap H_\omega^{-1}(E)$ have
density w.r.t.\ Liouville measure on $H_\omega^{-1}(E)$ scaling like
$E^{-(d-1)/2}$.

Thus for no total energy strictly above $2-d$
the distribution of asymptotic velocity is concentrated in zero.
\end{example}
%
\section{Poisson Potentials}\label{sec4}
%
Compared to the lattice case handled above, potentials based on marked
Poisson fields have some new features like unboundedness and invariance
w.r.t.\ $\bR^d$--translations.
So we discuss them in this section.
Many properties should generalise to other ergodic random potentials (like
gaussian potentials).

Again we start with single site potentials~$W_j$ indexed by~$j\in J$,
$\setsize J<\infty$,
but we assume for simplicity that for some~$\eta\ge2$
\[W_j\in C_c^\eta(\bR^d,\bR)\text.\]
\nomenclature[AWj]{$W_j$}{compactly supported single site potential}{}{}%
We now consider the marked Poisson process with space
\begin{align*}
  \tilde{\Omega}
    :=\bigl\{\omega\bigm|\omega&\text{ measure on
      $\bigl(\bR^d\times J,\ \cB(\bR^d\times J)\bigr)$ with}\\&
      \omega(K)\in\bN_0\text{ if $K\subseteq\bR^d\times J$ is compact}
    \bigr\}\text.
\end{align*}
\nomenclature[GOmegatilde]{$\tilde\Omega$}{counting measures on $\bR^d\times J$}{}{}%
$\tilde{\Omega}$ is the space of all counting measures on $\bR^d\times J$
and carries the vague topology generated by the basis consisting of the sets
\beq
\hspace*{-2mm}\cN_\omega(\psi_{0},\dotsc,\psi_k):=
      \bigl\{\omega'\in\tilde{\Omega}\bigm|\forall i\in\{0,\dotsc,k\}\colon
      \abs{\textstyle\int\psi_i\, d \omega'-\int\psi_i\, d \omega}<1\bigr\}
\Leq{eq:cN}
\nomenclature[ANomega]{$\cN_\omega$}{basis for vague topology}{}{}%
with $k\in\bN_{0}$, $\omega\in\tilde{\Omega}$ and $\psi_i\in C_c(\bR^d\times J,\bR)$
  ($i\in\{0,\dotsc,k\}$), see, \emph{e.g.}\ \cite{Rue87}.
  For $j\in J$ and compact $K\in\cB(\bR^d)$ the random variables
\[N_{K,j}\colon\tilde{\Omega}\to\bN_{0}\qtextq{,}
    \omega\mapsto\omega(K\times \{j\})\text,\]
\nomenclature[ANKj]{$N_{K,j}$}{particle number function}{}{}%
are called \emph{particle number functions}.
We fix intensities $\rho_j{}\ge0,\ j\in J$%
\nomenclature[Grhoj]{$\rho_j$}{intensity of Poisson process}{}{}.
Then $\beta$%
\nomenclature[Gbeta]{$\beta$}{Poisson measure}{}{}
is the unique probability measure on 
$\bigl(\tilde{\Omega},\cB(\tilde{\Omega})\bigr)$ with
\beq
    \beta\bigl(\{\omega\in\tilde{\Omega}\mid N_{K,j}(\omega)=m\}\bigr)
      =\frac{\bigl(\rho_j\lambda^d(K)\bigr)^m}{m!\exp\bigl(\rho_j\lambda^d(K)\bigr)} 
\Leq{beta:def}
  for all $m\in\bN_{0}$, 
  $j\in J$ and $K\in\cB(\bR^d)$ with $\lambda^d(K)<\infty$.

$\bigl(\bR^d\times J,\tilde{\Omega},\cB(\tilde{\Omega}),\beta\bigr)$ is called
\emph{marked Poisson process on $\bR^d$
with marks in~$J$ and intensities~$\rho_j$}, see \cite[Chap.~4.2]{SKM87}.
It induces the random potential
\[V\colon\tilde{M}:=\tilde{\Omega}\times\bR^d \longrightarrow\bR\qtextq{,}
  (\omega,q) \longmapsto \int_{\bR^d\times J} W_{j}(q-x)\, d\omega(x,j)\text.\]
\nomenclature[AV]{$V$}{Poisson potential}{}{}%
\begin{proposition}
  The potential $V\colon\tilde{\Omega}\to\bR$ is continuous, and
  \[V_\omega\in C^\eta(\bR^d,\bR)\qquad (\omega\in\tilde{\Omega})\text.\]
  There is a $\beta$--measure zero subset $N\subseteq\tilde{\Omega}$%
\nomenclature[AN]{$N$}{exceptional set of poisson configurations}{}{}
  such that for $\Omega:=\tilde{\Omega}\setminus N$ and extended phase space 
  $P:=\Omega\times\bR^d\times\bR^d$ the restriction 
  $H:=\tilde{H}\rstr_P$ of the Hamiltonian function
  \[\tilde{H}\colon\tilde{P}\to\bR \qtextq{,}
  (\omega,p,q)\mapsto\eh\norm p^2+V(\omega,q)\]
  induces a continuous Hamiltonian flow
  \[\Phi\colon\bR\times P\to P\text.\]
\end{proposition}
\begin{proof}
\begin{itemize}[$\bullet$]
\item To show continuity of~$V$, we construct for $\vep>0$ a neighbourhood
  \begin{equation*}
    U(\omega,q)
     \,:=\,\cN_\omega(\psi_1,\ldots,\psi_k)\times B_{\delta}(q)
     \,\subseteq\,\tilde M
     \,=\,\tilde\Omega\times\bR^d
  \end{equation*}
  \nomenclature[ABdelta]{$B_\delta$}{ball of radius~$\delta$}{}{}%
  of $(\omega,q){}\in\tilde\Omega\times\bR^d$, such that  $\abs{V(\omega',q')-V(\omega,q)}<\vep$
  for all $(\omega',q')\in U(\omega,q)$.
  For radius $R:=\sup\{\norm x\mid x\in\Union_{j\in J}\supp(W_j)\}+1$
  only the finitely many Poisson points in
  \begin{equation*}
    {\cS_{R,q}}:=\supp(\omega)\cap(B_R(q)\times J)
      \equiv\{(q_1,j_1),\ldots,(q_k,j_k)\}
  \end{equation*}
  \nomenclature[AS]{$\cS_{R,q}$}{localised support of poisson configuration}{}{}%
  can contribute to $V(\omega,q)$, and their minimal distance
  \begin{equation*}
    d_{\min}:=\inf \{\norm{q_m-q_n}\mid 1\leq m,n\leq k,q_m\ne q_n\}
  \end{equation*}
  is positive.
  We will need the loci~$\cS_{R,q}^j(\omega)$ of the support of~$\omega$
  with index~$j\in J$, implicitely given by
  \begin{equation*}
    {\cS_{R,q}} = \Union_{j\in J}\cS_{R,q}^j(\omega)\times\{j\}\text,
  \end{equation*}
  and a common Lipschitz constant~$L$ for the single site potentials
  $(W_j)_{j\in J}$.
  
  With
  $\delta:=\min\bigl\{\frac{\vep}{2L\,\omega(\cS_{R,q})},
                      \frac{d_{\min}}3,1\bigr\}>0$
  we choose $A_j:=B_R(q)\setminus B_\delta(\cS_{R,q}^j(\omega))$ and,
  employing the Kronecker delta $\delta_{\cdot,\cdot}$
  and the positive part $(\cdot)_+:=\max\{0,\cdot\}$,
  functions $\psi_i\in C_c(\bR^d\times J,\bR)$ with
  \begin{equation*}
    \psi_i(x,j):=\begin{cases}
      \bigl(1-\dist(x,A_j)/\delta\bigr)_+
	&(i=0)\\
      \bigl(1-\delta_{j,j_i}\dist(x,B_\delta(q_i))/\delta\bigr)_+
	&(i\in\{1,\ldots,k\})
    \end{cases}
  \end{equation*}

  Thus $\int_{{\bR^d}}\psi_i\,d\omega=(1-\delta_{0,l})\omega\{(q_i,j_i)\}$,
  and by definition~\eqref{eq:cN} of $\cN_\omega(\psi_0,\ldots,\psi_k)$,
  all $(\omega',q')\in U(\omega,q)$ fulfil the relations 
  \begin{equation*}
    \norm{q'-q}\leq\delta\qtextq,\omega'(A_j\times\{j\})=0
    \qtextq{and}
    \omega'(B_\delta(q_i)\times\{j_i\})=\omega\{(q_i,j_i)\}\text,
  \end{equation*}
  $j\in J$.
  With this preparation we deduce
  \begin{align*}&
    \abs{V(\omega',q')-V(\omega,q)}\\&
      \le\sum_{j\in J}\Bigabs{
	\Bigl(\sum_{\tilde q\in\cS_{R,q}^j(\omega')}
	  W_j(q'-\tilde q)\omega'(\tilde q,j)\Bigr)
	-\sum_{\hat q\in\cS_{R,q}^j(\omega)}W_j(q-\hat q)\omega(\hat q,j)
	}\\&
      \le\sum_{j\in J}\sum_{\hat q\in\cS_{R,q}^j(\omega)}
        \Bigabs{\Bigl(\sum_{\tilde q\in\cS_{\delta,\hat q}^j(\omega')}
	  W_j(q'-\tilde q)\omega'(\tilde q,j)\Bigr)
	  -W_j(q-\hat q)\omega(\hat q,j)
	}\\&
      \le\sum_{j\in J}\sum_{\hat q\in\cS_{R,q}^j(\omega)}
        \sum_{\tilde q\in\cS_{\delta,\hat q}^j(\omega')}
          \abs{W_j(q'-\tilde q)-W_j(q-\hat q)}\omega'(\tilde q,j)\\&
      \le\sum_{j\in J}\sum_{\hat q\in\cS_{R,q}^j(\omega)}
        2L\delta\,\omega(\hat q,j)
      =2L\delta\,\omega(\cS_{R,q})
      \le\ve\text,
  \end{align*}
  since $\norm{q'-\tilde q-(q-\hat q)}
          \le\norm{q'-q}+\norm{\tilde q-\hat q}
          \le2\delta$.
\item This implies that $\tilde{H}\colon\tilde{P}\to\bR$, too, is continuous
  on $\tilde{P}:=\tilde{\Omega}\times\bR^{2d}$ and for all $\omega\in\tilde
  {\Omega}$ the Hamiltonian $\tilde{H}_\omega\in C^\eta(\bR^{2d},\bR)$.
  However, this only guarantees \emph{local} unique existence of the flow.
  We now set 
  \beq
    N:=\bigl\{\omega\in\tilde\Omega\bigm|\liminf_{\norm q\to\infty}V_\omega(q)+
    c\LA q\RA^2=-\infty\ \text{for\ all}\ c>0\bigr\}
  \Leq{N}
  For $\omega\in\tilde\Omega\setminus N$ and $E\in\bR$ there exists a $C>0$ with
  \[\sqrt{2(E-V_\omega(q))_+}\leq C\LA q\RA \qquad (q\in\bR^d)\text.\]
  Thus the solution of the initial value problem with energy~$E$
  exists for all times.
  Namely as
  \[\LA q_\omega(t)\RA\frac{d}{dt}\LA q_\omega(t)\RA\leq\norm{q_\omega(t)}\,
  \norm{p_\omega(t)}\leq C\LA q_\omega(t)\RA^2\text,\]
  $t\mapsto\norm{q_\omega(t)}$ is at most of exponential growth.
\item
  We show that $\beta(N)=0$.
  We set $W_{\max}:=\max\{\abs{W_j(q)}\mid(q,j)\in\bR^d\times J\}>0$, 
  ${\rm diam}:=\max_j{\rm diam\,}({\rm supp\, }W_j)$
  \nomenclature[AWmax]{$W_{\max}$}{supremum of single site potentials}{}{}
  and $N_{B_r(q),J}:=\sum_{j\in J}N_{B_r(q),j}$.
  $\beta(N)=0$ follows from $\lim_{R\to\infty}\beta(N_R)=0$ with 
  \beq
    N_R:=\bigl\{\omega\in\tilde\Omega\bigm|\min_{\norm q\le R}V_\omega(q)\le
      -R^2
      \bigr\}\text.
  \Leq{NR}
  Now for $V_\omega(q)\le -R^2$ to occur for any point 
  $q\in B_{\rm diam}(Q)$
  we must have
  \beq
    N_{B_{2 {\rm diam}(Q)},J}(\omega)\ge R^2/W_{\max}.
  \Leq{qQ}
  As $B_R(0)$ is ${\rm diam}$--spanned (see Walters \cite{Wal82}) 
  by a set ${\cal Q}\subseteq B_R(0)$ 
  of points $Q$ with cardinality $\setsize{{\cal Q}}=\cO(R^d)$, 
  \eqref{NR} follows from \eqref{qQ} and \eqref{beta:def}.
%
\item  
  Finally, continuity of the flow~$\Phi$
  follows from the theorem on continuous dependence of solutions on parameters.
\end{itemize}
\end{proof}

\begin{remark}[Comparison With Quantum Mechanics]
  We have
  \beq
    \bE_\beta\l(\abs{V_\centerdot(0)}^{\frac{d+\vep}{2}}\ri)\leq 
    (W_{\max})^{\frac{d+\vep}{2}}\, \bE_\beta(N_{B_R(0),J})^{\frac{d+\vep}{2}}
  \Leq{est}
  Estimate \eqref{est} should be compared with a criterion for essential 
  self-adjointness of the corresponding random Schr\"odinger operator
  $-\Delta+V$ on $L^2(\bR^d)$.
  This is generally true for $\vartheta$--ergodic~$V$
  if $\bE_\beta\bigl(\abs{V(0)}^{\tilde{d}}\bigr)<\infty$ for
  $\tilde{d}:=2\lceil\frac{d+1}{4}\rceil$, see \cite{PF92}, Theorem 5.1.
\end{remark}

Unlike in the lattice case the Poisson potential~$V$ is invariant w.r.t.\ a
faithful 
$\bR^d$--action on~$P$:
For $\ell\in\bR^d$ set 
\begin{equation}\label{poissonaction}
  \vartheta_\ell\colon\tilde{\Omega}\to\tilde{\Omega}\qtextq{,} 
  \vartheta_\ell(\omega)(A):=\omega(\{(\ell+x,j)\mid(x,j)\in A\})\text.
\end{equation}
\nomenclature[Gthetaell]{$\vartheta_\ell$}{$\bR^d$--action on $\tilde\Omega$}{}{}%
Then~$N$, defined in \eqref{N}, is $\vartheta$--invariant so that we get
the $\bR^d$ action
\[\Theta\colon\bR^d\times P\to P\qtextq{,}
(\ell,\omega,p,q)\mapsto(\vartheta_\ell(\omega),p,q-\ell)\]
\nomenclature[GTheta]{$\Theta$}{$\bR^d$--action on poissonian~$P$}{}{}%
on extended phase space, with
\[H\circ\Theta_\ell=H \qtextq{,} \mu\circ\Theta_\ell=\mu \qtextq{and}
\Phi^t\circ\Theta_\ell=\Theta_\ell\circ\Phi^t\]
for all $t \in\bR^d$ and $\ell\in\bR^d$.

In order to control the existence of the asymptotic velocities, we select
an arbitrary regular lattice $\cL\subseteq\bR^d$, \emph{e.g.}\ $\cL:=\bZ^d$, and set
$\Theta^\cL:=\Theta\rstr_{\cL\times P}$%
\nomenclature[GThetaL]{$\Theta^\cL$}{$\cL^d$--action on poissonian~$P$}{}{}.

Then, as in \cref{sec2} above, we consider the covering projection
\[\hat{\pi}\colon P\to\hat{P}:=P/\Theta^\cL\]
to the factor space $\tP$.
Like in \eqref{homeo}, $\tP$ is homeomorphic to
$(\Omega\times\bR^d\times\cD)/{\sim}$.

On $\tP$ we consider the measures $\tmu$ and $\tmu_E\ (E\in\bR)$, defined
like in \eqref{mu:hat}.

Again, by $\Theta$--invariance similar to \eqref{hat:p} we can define momenta
on $\tP$ by
\[\tp\colon\tP\to\bR^d\qtextq{,} \tp\circ\tpi=p\text.\]

But unlike in \cref{sec2}, the potential $V_\omega$
is $\beta$--a.s.\ unbounded, and thus $\tH^{-1}((-\infty,E])$
is compact only if the $W_j$ are non--negative.
Therefore we need the following lemma:
\begin{lemma}\label{lem:L}
For all $E\in\bR$ the measure $\tmu_E$ on $\tP$ is finite, and
\begin{equation}\label{eq:I}
  \int_\tP\norm{\tp}\,d\tmu_E<\infty.
\end{equation}
\end{lemma}
\begin{proof}
We estimate, setting $r(q):=\sqrt{2(E-V(q))_+}$,
\[  \tmu_E(\tP)
    =\bE_\beta\l(
      \int_{\cD}\int_{\bR^d}\idty_{B_{r(q)}}(p)\,dp\,dq
      \ri)
    =\tau_d\,\lambda^d(\cD)\,
    \bE_\beta \l( 2\bigl( E-V(0) \bigr)_+^{\ d/2}\ri)
    <\infty
\]
with $\tau_d:=\lambda^d\bigl(B_1(0)\bigr)=\frac{\pi^{d/2}}{\Gamma(1+d/2)}$,
\nomenclature[Gtaud]{$\tau_d$}{Lebesgue volume of unit ball in $\bR^d$}{}{}%
and similar (but with exponent $(d+1)/2$) for~\eqref{eq:I}.
\end{proof}

Thus we get the analog of \cref{prop:A}:
\begin{proposition}
The asymptotic velocities
\ $\ovv^\pm(\omega,x_0)\,:=\,
\lim_{T\to\pm\infty}\frac{q_\omega(T,x_0)}{T}$
exist and are equal $\beta$--a.s.\ on~$P$.
Furthermore 
\[\ovv\in L^\infty_{\mathrm{loc}}(P) \qtextq{and} \ovv\circ\Phi^t=
\ovv\circ\Theta_\ell=\ovv\text.\]
\end{proposition}
\begin{proof}
  We invoke Birkhoff's theorem for
  \begin{equation*}
    \ovv^\pm(\omega,x_0)
      =\lim_{T\to\pm\infty}\frac1T\int_0^T\tp\bigl(t,\tpi(\omega,x_0)\bigr)\,dt\text,
  \end{equation*}
  using \cref{lem:L}.
\end{proof}
\begin{remarks}
\begin{enumerate}[1.]
\item
  The $\bR^d$ action~\eqref{poissonaction} is mixing on $\tilde\Omega$,
  see \cite[p.~27]{PF92}.
  Thereby the action of the lattice~$\cL$ by translations
  $\vartheta_\ell\ (\ell\in\cL)$ on~$\tilde\Omega$ is mixing, too,
  and in particular $\beta$--ergodic.
  Thus, like in \cref{prop:B}, the energy--velocity distributions
  $\Gamma_\omega$ for the Poisson potentials are $\beta$--a.s.\ deterministic.
\item
  For $d=1$ one gets a result similar to \cref{prop:C}.
  Note, however that $\supp\bigl(V(\beta\otimes\lambda^d)\bigr)$
  equals the closure of $(V_{\ess\min},\infty)$,
  if there is a single site potential $W_j$ with $W_j(q)>0$ for some $q\in\bR$.
  In that case all motion is bounded $\beta$--a.s.\,.
\end{enumerate}
\end{remarks}
%
\section{Singular Values of the Hamiltonian}\label{sec4a}
%
Since Hamiltonian motion enjoys conservation of energy,
one has to decompose phase space into energy shells
in order to find ergodic motion.
This is possible for regular energy values,
so we first study the set of singular values of the Hamilton function.
Both for the lattice and the Poisson case we have the following result:

\begin{proposition}
The closure of the set of singular values of $V_\omega$ is 
$\beta$--almost surely deterministic.
\end{proposition}
\begin{proof}
For $S_{k,m}:=[m 2^{-k},(m+1) 2^{-k}]$ with $k\in\bN$ and $m\in\bZ$
the set
\[\Omega_{k,m}:=\bigl\{\omega\in\Omega\bigm|\exists q \in \bR^d:
\nabla_q  V(\omega,q) =0\mbox{ and } 
 V(\omega,q)\in  S_{k,m}\bigr\} \]
is $\cL$-invariant. Thus by $\beta$--ergodicity it is of measure zero or one.
The sets $S_k:=\bigcup_{m\in\bZ:\, \beta(\Omega_{k,m})=1} S_{k,m}$ are closed, and 
$S_{k+1}\subseteq S_k$.
\[\Omega_{k}:=\{\omega\in\Omega\bigm|{\rm CVal}_\omega\subseteq S_k \}\]
is still invariant and of measure one.
The same is true for $\Omega_{\infty}:= \bigcap_{k\in\bN}\Omega_{k}$.
The potentials indexed by $\omega\in \Omega_{\infty}$ have their critical values
in the closed set $S_\infty:=\bigcap_{k\in\bN}S_k$.
It is $\beta$--almost surely the closure of the set of singular values of~$V_\omega$.
\end{proof}

\begin{example}[Denseness of Singular Values]
The set of singular values may be dense in 
$[V_{\ess\min},V_{\ess\max}]$.
This is the case
$\beta$--almost surely for the random $\bZ$--lattice potential on 
$\bR$ given by the single site potentials 
\[W_j(q):=\textstyle{ \sum_{\ell\in\bN}}\; j
\abs J^{-\ell}\chi\bigl(q-(-1)^\ell\lfloor\ell/2\rfloor\bigr)
\quad\bigl(q\in\bR,\;j\in J:=\{0,\ldots,\setsize J-1\}\bigr)\text,\] 
with $\chi\in C^\infty_c(\bR,[0,1])$, $\chi\rstr_{\l[-\ev,\ev\ri]}=1$
and ${\rm supp} (\chi)\subseteq[-\eh,\eh]$, if $\beta:=\otimes_{\bZ} \hat{\beta}$ 
is the product measure of $\hat{\beta}:=\textstyle{
\sum_{j\in J}}\delta_j/\setsize J$. 
Then for any $\omega\in \Omega=J^\bZ$ with a dense
$\vartheta$--orbit the set of singular values is dense in $[0,1]$\, 
and this is the case for $\beta$-a.e.\ $\omega\in\Omega$.\par 
Note that the fall-off of the
single site potentials is not only polynomial as assumed in 
\eqref{as:sisi} but exponential, with rate 
$W_j(q)=\cO\bigl(\setsize J^{-2|q|}\bigr)$.
\end{example}
The exponential decay in this example is below
the rate that ensures measure zero for the closure of the set 
${\rm CVal}_\omega =V_\omega ({\rm CSet}_\omega)$%
\nomenclature[ACVal]{$\mathrm{CVal}$}{set of singular values of~$V$}{}{}%
\nomenclature[ACSet]{$\mathrm{CSet}$}{set of singular points of~$V$}{}{}
of singular values of $V_\omega$:

\begin{proposition}
  Let for $d=1$ and some $\vep>0$
  the single site potentials $W_j\in C^2(\bR,\bR)$
  obey (for the lattice $\bZ$)
  the decay estimate $W_j'(q)=\cO\bigl(\setsize J^{-(4+\vep)\abs q}\bigr)$.
  Then
  \[\lambda^1\bigl(\overline{{\rm CVal}_\omega}\bigr)=0
  \qquad(\omega\in\Omega)\text.\]
\end{proposition}
\begin{proof}
$\bullet$
By using the lattice translation $\vartheta$, 
it suffices to show that
\beq
\lambda^1\left(\overline{
{\textstyle \bigcup_{\omega\in\Omega}} 
V_\omega ({\rm CSet}_\omega \cap[0,1]) }\right)=0.
\Leq{Vla0}
For this we prove that for $R\in\bN$ large
and the middle sequence $\omega^{(R)}%
\nomenclature[GomegaR]{$\omega^{(R)}$}{middle sequence of $\omega$}{}{}
\in J^{\{-R,\ldots,R\}}$
of $\omega$ we have
\beq
V_\omega \bigl({\rm CSet}_\omega \cap[0,1] \bigr)\;\subseteq\;
M_{\omega^{(R)}}\qquad (\omega\in\Omega)\text,
\Leq{M:inclusion}
for suitable closed subsets $M_{\tau} \subseteq \bR$
of measures  
\beq
\lambda^1(M_{\tau})\le \setsize J^{-(2+\vep/4)R}\qquad 
\bigl(\tau\in J^{\{-R,\ldots,R\}}\bigr)\text.
\Leq{M:measure}
\nomenclature[Gtau]{$\tau$}{finite sequence in~$J$}{}{}%
As this implies 
$\lambda^1\Bigl(\overline{
{\textstyle \bigcup_{\omega\in\Omega} 
V_\omega ({\rm CSet}_\omega \cap}[0,1])}\Bigr)\le \setsize J^{1-\vep R/4}$,
\eqref{Vla0} then follows as $R\ar\infty$. \\
$\bullet$
$c^{(2)}:=\sup\bigl\{\sum_{j\in\bZ}\abs{W''_{\omega_j}(q-j)}\bigm|
\omega\in\Omega,q\in\bR\bigr\}<\infty$ is a uniform upper bound for the second derivative of
\[\tilde{V}_\tau(q) := 
{\textstyle \sum_{\ell=-R}^R} W_{\tau(\ell)}(q-\ell)\qquad
\bigl(q\in\bR,\,\tau\in J^{\{-R,\ldots,R\}}\bigr)\text.\]
Then for $\delta\le c^{(2)}$ the disjoint union 
$\{x\in [0,1]\mid\abs{\tilde{V}'_\tau(x)}\le \delta\}$ of closed intervals 
has at most $i_{\max}(\tau)\le c^{(2)}/\delta$ intervals $I_1(\tau),\dotsc,I_{i_{\max}(\tau)}(\tau)$
containing a point~$q$ with $\tilde{V}'_\tau(q)=0$,
since each such interval has at least length $\delta/c^{(2)}$.
\\[2mm]
We set 
$M_\tau:=\bigcup_{i=1}^{i_{\max}(\tau)} B_{\delta^2}\big(\tilde{V}_\tau(I_i(\tau))\big)$%
\nomenclature[AUdelta]{$B_\delta$}{open neighbourhood of size~$\delta$}{}{}.
This union of thickened intervals has measure 
\[\lambda^1(M_\tau)\le
\lambda^1\Bigl(\tilde{V}_\tau\big({\Union\nolimits}_{i=1}^{i_{\max} 
(\tau)}I_i(\tau)\big)\Bigr)
+ 2\delta^2\cdot i_{\max}(\tau)\le (1+2c^{(2)})\;\delta\text.\]
So for $\delta\equiv\delta_R:= \setsize J^{-(2+\vep/2)R}$
and $R$ large, \eqref{M:measure} is satisfied.\\
$\bullet$
By the decay assumption $|W_j'(q)|\le c \setsize J^{-(4+\vep)\abs q}$
on the derivatives of the single site potentials, for 
$q\in[0,1]\setminus \bigcup_{i=1}^{i_{\max}(\tau)} I_i(\tau)$ one has
\begin{equation*}
  \abs{V'_\omega(q)}
    \ge\abs{\tilde{V}'_\tau(q)}-c \sum_{\ell\in\bZ:|\ell|>R}\setsize J^{-(4+\vep)\abs {\ell}}
    \ge
\delta_R-c\,\delta_R^2 >0
\end{equation*}
for $R$ large enough. So 
 ${\rm CSet}_\omega \cap[0,1]\subseteq \bigcup_{i=1}^{i_{\max}(\tau)} I_i(\tau)$
in \eqref{M:inclusion}.\\
$\bullet$ 
The single site potentials themselves, and not only their derivatives, 
decay like $W_j(q)=\cO\bigl(\setsize J^{-(4+\vep)|q|}\bigr)$.
This means that 
\[\abs{V_\omega(q)- \tilde{V}_{\omega^{(R)}}(q)}\le
c \textstyle \sum_{\ell\in\bZ:|\ell|>R}\setsize J^{-(4+\vep)\abs {\ell}} \le\delta^2
\qquad\big(q\in[0,1]\big).\]
But $\delta^2$ is the parameter appearing in the definition of $M_\tau$.
So \eqref{M:inclusion} is satisfied, too.
\end{proof}
A similar statement, but with 
superexponential decay rate $\cO\bigl(\setsize J^{-(c_\cL+\vep)\|q\|^d}\bigr)$, 
should be sharp for $d\ge 1$,
with a constant $c_\cL$ depending on the lattice
$\cL\subseteq\bR^d$.

We denote the restriction of the flow $\Phi_\omega$ to the energy surface
$\Sigma_{E,\omega}:=H_\omega^{-1}(E)\subseteq\bR^{2d}$%
\nomenclature[GSigmaEomega]{$\Sigma_{E,\omega}$}{energy surface}{}{}
by $\Phi_{E,\omega}$%
\nomenclature[GPhiEomega]{$\Phi_{E,\omega}$}{flow on~$\Sigma_{E,\omega}$}{}{}.
For regular values~$E$ of $H_\omega$,\ $\Sigma_{E,\omega}$ carries a 
$\Phi_{E,\omega}$--invariant measure $\lambda_{E,\omega}$%
\nomenclature[GlambdaEomega]{$\lambda_{E,\omega}$}{Liouville measure on $\Sigma_{E,\omega}$}{}{}
derived from Lebesgue measure $\lambda^{2d}$ on phase space 
(see, \emph{e.g.}\ \cite{AM78}, Thm.~3.4.12).

This non--atomic \emph{Liouville measure} thus exists for all 
$E>V_{\max}$, using the equality of regular values of $H_\omega$ and $V_\omega$.

For $V_\omega\in C^{d}(\bR^d,\bR)$ 
by Sard's theorem (see, \emph{e.g.} \cite[Thm. 3.1.3]{Hir76}) it  
in fact exists for almost all energies $E>\inf V_\omega$.
\begin{remark}[Denseness of Singular Values, Poisson Case]
For (non--trivi\-al) Poisson potentials, 
$\beta$--almost surely,
the set $\overline{{\rm CVal}_\omega}$ equals $\bR$, $[0,\infty)$ or  
$(-\infty,0]$, depending on the signs occuring in
the union of ranges of the single site potentials.
This is true, since
\begin{enumerate}[(a)]
\item 
for any real number $r$ between zero and an extreme value of $W_j$
the sum $U_a(q):=W_j(q)+W_j(q+ae_1)$%
\nomenclature[AUa]{$U_a$}{sum of two Poisson single site potentials}{}{}
of~$W_j$ and its translate by a suitable $a={}a(r)\in\bR$ has~$r$
as a singular value.  For the case $\max\; W_j=W_j(q_0)>0$ we take 
\[ \mathrm{MM}(a) := 
\sup_{c\in C} \min_{t\in [0,1]} U_a\bigl(c(t)\bigr)\text,\] 
\nomenclature[AMM]{$\mathrm{MM}$}{min--max function}{}{}%
with $C := 
\bigl\{c\in C^\infty\bigl([0,1],\bR^d\bigr) \bigm|c(0)=q_0, c(1)=q_0+a e_1\bigr\}$.
Then $\mathrm{MM}(0)=2\max W_j$, $\mathrm{MM}(a)=0$ for $\abs a$~large, and 
$\mathrm{MM}$ is continuous.
Furthermore $\mathrm{MM}(a)$ is a singular value for~$U_a$.

The case $\min W_j=W_j(q_0)<0$ is treated by the reverse minmax problem.
\item 
Integer multiples of $W_j$ are Poisson potentials, too.
\end{enumerate}
\end{remark}
%
\section{Notions of Ergodicity of Time Evolution}\label{sec5}
%
In general $\lambda_{E,\omega}$ is non--finite.
Nevertheless, we may ask whether it is an \emph{ergodic} measure (in the 
sense of Aaronson \cite{Aar97}), that is,
whether the only $\Phi_\omega$--invariant measurable subsets
$A\subseteq\Sigma_{E,\omega}$ have measure zero or full measure,
i.e.\ $\lambda_{E,\omega}(\Sigma_{E,\omega}\setminus A)=0$.
Instead, we may also ask whether the restriction
\beq
\tPhi_E^t:=\tPhi^t{\rstr}_{\hS_E}\ (t\in\bR)
\Leq{tPhiE}
\nomenclature[GPhihatEt]{$\tPhi_E^t$}{flow on $\hS_E$}{}{}%
of the flow $\tPhi$ to the (extended) energy surfaces
$\hS_E:=\tH^{-1}(E)\subseteq\tP$%
\nomenclature[GSitmahatE]{$\hS_E$}{finite energy shell $\hS_E$}{}{}
is ergodic.
Here the Liouville measure $\tlambda_E$%
\nomenclature[GlambdahatE]{$\tlambda_E$}{Liouville measure on~$\hS_E$}{}{}
on $\hS_E$, derived from the measure $\tmu$ on $\tP$ by disintegration,
is finite.  Ergodicity of $\hat\Phi_E$ is much cheaper obtained
than ergodicity of $\Phi_{E,\omega}$,
since on finite measure spaces Poincar\'e{'s} recurrence theorem holds,
see \cref{sect10}.
In this section we clarify the relations
between the two notions of ergodicity
and give some immediate consequences.
The results pertain to both lattice and Poisson potentials
(using quotients by $\bZ ^d$ in the latter case).
\begin{proposition}\label{prop:erg:erg}
  If for energy $E\in\bR$, the flow $\Phi_{E,\omega}$ on $\Sigma_{E,\omega}$
  is $\lambda_{E,\omega}$--ergodic for $\beta$--a.e.\ $\omega\in\Omega$, then
  the flow $\tPhi_E$ on $\Sigma_E$ is $\tmu_E$--ergodic.
\end{proposition}
\begin{proof}
  Assume that $\hat\Phi_E$ is not ergodic.
  Then there exists a measurable invariant
  ($\tlambda_E\bigl(\hat{A}\Delta\tPhi_E^t(\hat{A})\bigr)=0\ \text{for\ all}\ t\in\bR$)
  subset $\hat{A}\subseteq\hS_E$ with 
  $0<\tlambda_E(\hat{A})< \tlambda_E(\hS_E)$.
  Its pre-image $A:=\tpi^{-1}(\hA)\subseteq\Sigma_E=H^{-1}(E)$ is measurable
  and invariant, too, w.r.t.\ the Liouville measure $\lambda_E$ 
  on $\Sigma_E$.
  Finally, $\lambda_E(A)>0$ and 
  $\lambda_E(A^c)>0$, as $A^c=\Sigma_E\setminus A=
  \tpi^{-1}\bigl(\hS_E\setminus\hA\bigr)$.
  \par
  So $A_\omega:=A\cap\Sigma_{E,\omega}$ is $\Phi_{E,\omega}$--invariant
  $\beta$--a.s., and similar for $A_\omega^c:=A^c\cap\Sigma_{E,\omega}=
  \Sigma_{E,\omega}\setminus A_\omega$.
  \par
  By $\vartheta$--ergodicity of $\beta$ the $\beta$-a.s.\ $\vartheta$-invariant functions
  \begin{equation*}
    f_A,f_{A^c}\colon\Omega\to\{0,1\} \qtextq{,} f_A(\omega)=
    \begin{cases}
      1&\text,\ \lambda_{E,\omega}(A_\omega)>0\\
      0&\text{, otherwise}
    \end{cases}
  \end{equation*}
  (resp.\ for $A_\omega^c$) are both $\beta$--a.s.\ constant with value~$1$.
  Thus $\Phi_{E,\omega}$ is not $\lambda_{E,\omega}$--ergodic $\beta$--a.s.\,.
\end{proof}
Note in passing that the flow $\Phi_E$ on $\Sigma_E$ is \emph{not} 
$\lambda_E$--ergodic (given $\beta$ is not a Dirac measure on $\Omega$),
since $\Omega$ is unchanged by that flow.

Next we ask about the dynamical consequences of ergodicity.
\begin{proposition}\label{prop:P}
  Given $E\in\bR$ and $\omega\in\Omega$ such that $\Phi_{E,\omega}$ is $\lambda_{E,\omega}$--ergodic
  on the regular energy surface $\Sigma_{E,\omega}$,
  \begin{compactitem}
    \item then the asymptotic velocity $\ovv(x)=0$
      for a.e.\ $x\in{}\Sigma_{E,\omega}$,
    \item but the motion is unbounded a.e.\ on $\Sigma_{E,\omega}$
      unless $\Sigma_{E,\omega}$ is compact.
  \end{compactitem}
\end{proposition}
\begin{remark}
  Note that $\Sigma_{E,\omega}$ is non--compact for $E>V_{\max}$
  and $\beta$--a.s.\ non--compact for $E>V_{\ess\min}$, given ergodicity of~$\Phi_{E,\omega}$.
\end{remark}
\begin{proof}
Remember that $\ovv(x)=0$ by definition if $\ovv^\pm(x)$ do not exist or are
unequal.
\begin{itemize}[$\bullet$]
  \item Assume that $\lambda_{E,\omega}(\{x\in\Sigma_{E,\omega}\mid
    \ovv(x)\neq0\})>0$.
    Then there exists an open half space $H_s:=\{q\in\bR^d\mid\LA q,s\RA>0\}$%
    \nomenclature[AHs]{$H_s$}{half space}{}{}
    indexed by a unit vector $s\in S^{d-1}$%
    \nomenclature[ASd]{$S_d$}{$d$--dimensional sphere}{}{}
    so that
    \[\lambda_{E,\omega}(B_s)>0\qtextq{for}
      B_s:=\{x\in\Sigma_{E,\omega}\mid\ovv(x)\in H_s\}\text.
    \]
    By reversibility of the flow and $\ov v^-=\ov v^+$
    \[\lambda_{E,\omega}(B_{-s})=\lambda_{E,\omega}(B_s)\text.\]
    On the other hand, $\ovv\circ\Phi^t=\ovv\quad(t\in\bR)$, so that $B_s$ and 
    $B_{-s}$ are disjoint invariant subsets of positive measure.
    Thus $\Phi_{E,\omega}$ is not $\lambda_{E,\omega}$--ergodic.
  \item Assume that there exists a subset $B\subseteq \Sigma_{E,\omega}$ of 
    points leading to forward bounded motion with $\lambda_{E,\omega}(B)>0$.
    Then for some $r>0$ the subset 
    \[B_r:=\{(p,q)\in B\mid\norm q\leq r\}\] 
    is of positive measure, too.
    Furthermore, for some $r'>r$ the same applies to
    \[B_{r,r'}:=\{x\in B_r\mid\norm{q_t(\omega,x)}\leq r'
    \text{ for\ all}\ t\ge0\}\text.\]
    Finally the set
    \[\tiB:=\bigcup_{t\ge0}\Phi_{E,\omega}^t(B_{r,r'})\subseteq \Sigma_E\]
    is $\Phi_{E,\omega}$--invariant w.r.t.\ $\lambda_{E,\omega}$, and 
    $\lambda_{E,\omega}(\tiB)>0$.
    But by non--compactness of~$\Sigma_{E,\omega}$,
    its complement $\tiB^c=\Sigma_{E,\omega}\setminus\tiB$
    is of measure $\lambda_{E,\omega}(\tiB^c)>0$, since~$\tiB$
    is bounded: $\norm q\leq r'$ if $(p,q)\in\tiB$.
    Thus $\Phi_{E,\omega}$ is not $\lambda_{E,\omega}$--ergodic.
\end{itemize}
\end{proof}

In one--dimensional natural mechanical systems with Hamiltonian
$H\colon\bR^2\to\bR,\ H(p,q)=\frac{1}{2}p^2+V(q)$ ergodicity on regular energy 
surfaces $\Sigma_E=H^{-1}(E)$ is widespread and occurs for compact as well 
as non--compact $\Sigma_E$ (examples: $V(q)=q^2,\ V(q)=q$).

In our present context, however, ergodicity in one dimension is exceptional.
\begin{proposition}
For $d=1$ and $E>V_{\ess\min}$ the flow $\Phi_{t,\omega}$ on a
regular energy surface $\Sigma_{E,\omega}$ is $\beta$--a.s.\ not 
$\lambda_{E,\omega}$--ergodic.
\end{proposition}
\begin{proof}
  We assume that $\omega$ is chosen so that $\ovv^+(x)=\ovv^-(x)$ for
  $\lambda_{E,\omega}$--a.e.\ $x\in\bR$.
  This assumption is valid $\beta$--a.s.\,.
  For $E<V_{\ess\max}$ the motion is bounded $\beta$--a.s., whereas for
  $E>V_{\ess\max}$ asymptotic velocity is non--zero
  $\beta$--a.s.\ according to \cref{prop:C}.

  Both statements contradict \cref{prop:P} (note that
  $\Sigma_{E,\omega}$ is non--compact $\beta$--a.s.\ for $E>V_{\ess\min}$).

  But for the case left, $E=V_{\ess\max},\ \Sigma_{E,\omega}$ must be 
  connected, regular and non--compact.
  Thus it must be diffeomorphic to $\bR$.
  This can only occur if there is a $q_0$ so that $V(q)>E$ for all $q\in
  (q_0,\infty)$ and $V(q)<E$ for all $q\in(-\infty,q_0)$ or vice versa.
  This, however, would contradict our initial assumption $\ovv^+(x)=\ovv^-(x)$.
\end{proof}
\Cref{ex:la} together with \cref{prop:P} shows that for any
dimension~$d$ there are non--trivial random lattice potentials that lead to
non--ergodic motion on the regular energy surfaces $\Sigma_{E,\omega}$ for
\emph{any}~$E$ and $\omega\in\Omega$.

The statements of \cref{prop:P} can be proven under a different
assumption, which by \cref{prop:erg:erg} is weaker than
$\lambda_{E,\omega}$--ergodicity assumption in \cref{prop:P}
\emph{for $\beta$-a.e.}~$\omega$:
%
\begin{proposition}\label{prop:ConsequencesOfErgodicity}
  If $\tPhi_E$ is $\tmu_E$--ergodic on the regular energy surface $\hS_E$,
  \begin{compactitem}
    \item then the asymptotic velocity $\ovv(x)=0$ for a.e.\ $x\in{}\Sigma_{E,\omega}$,
    \item but the motion is unbounded a.e.\ on $\Sigma_{E,\omega}$
  \end{compactitem}
  for $\beta$-a.e.\ $\omega$.
\end{proposition}
\begin{proof}
  Still, $\ovv(x)=0$ by definition,
  where $\ovv^\pm(x)$ do not exist or are unequal.
  We define $\hat v_E\colon\hS_E\to\bR^d$ by
  $\hat v_E\circ\hat\pi=\ovv\rstr_{\Sigma_E}$.
%
\begin{itemize}[$\bullet$]
  \item Assume that there exists a measurable $\Omega'\subseteq\Omega$
    with $\beta(\Omega')>0$ such that
    $\lambda_{E,\omega}(\{x\in\Sigma_{E,\omega}\mid\ovv(x)\ne0\})>0$
    for all $\omega\in\Omega'$.
    This implies
    \begin{equation*}
      \tlambda_E(\{x\in\hS_E\mid\hat v_E(x)\ne0\})>0\text.
    \end{equation*}
    Again, there exists an open half space $H_s:= \{q\in\bR^d\mid\LA q,s\RA>0\}$
    indexed by a unit vector $s\in S^{d-1}$ so that
    \begin{equation*}
      \tlambda_E({\hat B}_s)>0\qtextq{for}
      {\hat B}_s:=\{x\in\hS_E\mid\hat v_E(x)\in H_s\}\text.
    \end{equation*}
    By reversibility of the flow and $\ov v^-=\ov v^+$
    \[\tlambda_{E,\omega}({\hat B}_{-s})=\tlambda_{E,\omega}({\hat B}_s)\text.\]
    On the other hand, $\hat v_E\circ\tPhi^t=\hat v_E\quad(t\in\bR)$,
    so that ${\hat B}_s$ and ${\hat B}_{-s}$
    are disjoint invariant subsets of positive measure.
    Thus $\tPhi_E$ is not $\tlambda_E$--ergodic.
  \item We have to show, that the bounded orbits in~$\Sigma_E$
    carry no measure.  Fix $r>0$ and let
    \begin{equation*}
      B{{}\equiv B_r}:=\{x\in\Sigma_E\mid\forall t\in\bR\colon\norm{q_t(x)}<r\}
    \end{equation*}
    be the set of orbits bounded by~$r$.
    With help of the projections
    \begin{equation*}
      \pi_{\Omega,E}\colon\Sigma_E\to\Omega
      \qtextq{and}
      \pi_\Lambda\colon\Omega\to\Omega_\Lambda:=J^\Lambda
    \end{equation*}
    \nomenclature[GpiOmegaE]{$\pi_{\Omega,E}$}{projection $\Sigma_E\to\Omega$}{}{}%
    \nomenclature[GpiLambda]{$\pi_\Lambda$}{projection $\Omega\to\Omega_\Lambda$}{}{}%
    for finite subsets $\Lambda\subseteq\cL$%
    \nomenclature[GLambda]{$\Lambda$}{finite subset in~$\cL$}{}{}
    we can partition $B=\Union_{\sigma\in\Omega_\Lambda}B_\sigma$ with
    \begin{equation*}
      B_\sigma:=B\isect(\pi_\Lambda\circ\pi_{\Omega,E})^{-1}\{\sigma\}
      \qquad(\sigma\in\Omega_\Lambda)\text.
    \end{equation*}
    Since the probability measure~$\beta$ is non--atomic, so is $\lambda_E$.
    We thereby conclude
    \begin{equation}\label{eq:Bsigmatozero}
      \sup_{\sigma\in\Omega_\Lambda}\lambda_E(B_\sigma)
        \xto{\Lambda\to\cL}0\text.
    \end{equation}
    On the other hand $B_\sigma$ is $\Phi_E$--invariant,
    and $\hat B_\sigma:=\hat\pi(B_\sigma)$ is preserved by $\tPhi_E$.
    By ergodicity of~$\tPhi_{E}$ and \eqref{eq:Bsigmatozero}
    we find a finite subset $\Lambda\subseteq\cL$ such that
    \begin{equation*}
      \sup_{\sigma\in\Omega_\Lambda}\tlambda_E(\hat B_\sigma)=0\text.
    \end{equation*}
    But $\hat\pi$ is non--singular, i.e.\ preserves sets of measure~$0$,
    and we conclude
    \begin{equation*}
      \lambda_E(B)
        =\sum_{\sigma\in\Omega_\Lambda}\lambda_E(B_\sigma)
	=0\text.
    \end{equation*}
\end{itemize}
\end{proof}

%
\section{No Ergodicity in the Poisson Case}\label{sec5a}
%
We now assume for the Poisson case that the single site potentials are smooth,
in addition to being compactly supported.
Then ergodicity of the dynamics is atypical:
\begin{theorem}
\label{thmPoisson}
Consider a random Poisson potential on $\bR^d$.
Then for any $E\in\bR$ the motion on the energy surface
$\Sigma_{E,\omega}$ is $\beta$--a.s.\ not ergodic.
\end{theorem}

\noindent
{\bf Strategy of proof and first steps}\\
$\bullet$
The theorem is true for $d=1$ dimensions. If there is a single site 
potential, say $W_1$, and $q\in\bR$ with $W_1(q)>0$, then we have
bounded orbits $\beta$--a.s.\ for all~$E$. If instead all $W_j$ are
non--positive, $\Sigma_{E,\omega}$ has two connected components for 
$E>0$ and one has $\beta$--a.s.\ bounded orbits for $E\le 0$.\\
$\bullet$
So we assume $d\ge2$.
If the Poisson potential is zero, then we have free motion,
which is not ergodic in any dimension. 
Otherwise there is a non--zero
single site potential, say $W_1$. 
The proof method depends on the sign of 
\beq
  I:={\textstyle \int_{\bR^d}} W \,d\lambda^{d}
\Leq{def:I}
\nomenclature[AI]{$I$}{integral over $W_1$}{}{}%
(we temporarily omit the index of $W_1$).
Then, for given $I$ and energy~$E$, we construct a set~$\cL$%
\nomenclature[AL]{$\cL$}{set of translation points for poisson constructions}{}{}
and thereby a \emph{finite} sum $q\mapsto \sum_{\ell\in\cL} W(q-\ell)$  
of translated potentials which for the given
energy confines trajectories of positive measure.\\
This then suffices to show that $\beta$--a.s.\
the flow on $\Sigma_{E,\omega}$ is not ergodic:

\begin{enumerate}[1.]
\item
If the restriction of $\omega$ 
to a large ball $B_{R}(0)\subseteq\bR^d$ is near to 
$\sum_{\ell\in \cL}\delta_{(\ell,1)}$ w.r.t.\ the
topology from \eqref{eq:cN}, then the flow is shown to be
non--ergodic, too. This event 
in $\Omega$ has positive probability.
\item
By the nature of the Poisson process $\beta$
the Borel--Cantelli Lemma can be applied to the 
$\vartheta_\ell$--translates of that
event with $\ell\in 2R\bZ^d$, 
so that the flow is even $\beta$--a.s.\ non--ergodic.
\end{enumerate}

\nomenclature[ASOd]{$\mathrm{SO}(d)$}{special orthogonal group}{}{}%
${\rm SO}(d)$--invariant potentials
$U\in C^\infty_c\bigl(\bR^d,[0,\infty)\bigr)$%
\nomenclature[AU]{$U$}{$\mathrm{SO}(d)$--invariant function}{}{}
lead to integrable motion, and for suitable~$U$ to a positive measure 
of bounded orbits 
(see, {\em e.g.}, Arnol'd  \cite{Arn78}, Section 8). \\
We first approximate a given such
$U\in C^\infty_c\bigl(\bR^d,[0,\infty)\bigr)$, 
with $U(q)=0$ for~$q$ near zero,
by summing scaled translates of the single site function 
$W\in C^\infty_c(\bR^d,\bR)$.
We have $\int_{\bR^d} W_a \,d\lambda^{d}=I$ for
\nomenclature[AWa]{$W_a$}{$W$ scaled by~$a$}{}{}%
\begin{equation}\label{W-scaling}
  W_a(q) := a^{-d} W(q/a)\qquad (a>0)\text.
\end{equation}%

\begin{lemma}\label{lem:dens}
  For $I>0$ in \eqref{def:I} and $U$
  as above there exists a map $g\in C^{\infty}(\bR^d,\bR^d)$ 
  which, restricted to the interior of ${\rm supp}(U)$, 
  is a diffeomorphism onto its (bounded) image $\mathrm{Im}$%
  \nomenclature[AIm]{$\mathrm{Im}$}{image of~$U$}{}{}
  such that the distributions 
  \[D_\vep\;:= \;\vep^d \;
  {\textstyle \sum_{\ell\in\vep \bZ^d \cap \mathrm{Im}}} 
  \;\delta_{g^{-1}(\ell)} \qquad(\vep>0)\]
  \nomenclature[ADepsilon]{$D_\vep$}{distribution approximating $D_0$}{}{}%
  \nomenclature[Gdeltax]{$\delta_x$}{Dirac measure on~$x$}{}{}%
  converge vaguely to the distribution $D_0:=U\lambda^d$%
  \nomenclature[ADzero]{$D_0$}{distribution with density~$U$}{}{}
  with density~$U$ w.r.t.\ Lebesgue measure $\lambda^d$ for $\vep\searrow0$.
  Moreover, in uniform $C^k$ norm, with $c:=1/(2(d+k+1))$,
  \begin{equation}\label{faltung}
    I^{-1}\lim_{\vep\searrow 0}D_\vep * W_{\vep^c}\; =\; U.
  \end{equation}
\end{lemma}
\begin{proof}
Write $U(q)=\tilde U(\norm q)$ and set 
$\tilde g(r) := \bigl(d\int_0^r x^{d-1}\tilde U(x)\,dx\bigr)^{1/d}$.

Then, by the assumptions on $U$, the map 
$g\colon\bR^d\ar\bR^d$, with $g(0):=0$ and $g(q):=\tilde g(\|q\|)
\frac{q}{\norm q}$ else, is zero in a neighbourhood of zero, 
$g\in C^k$ and bounded.

As $Dg(q)=\tilde g(\norm q)(\idty-P_q)/\norm q+\tilde g'(\norm q)P_q$
with the orthogonal projection~$P_q$%
\nomenclature[APq]{$P_q$}{orthogonal projection}{}{}
onto ${\rm span}(q)$%
\nomenclature[Aspan]{$\mathrm{span}$}{linear span}{}{},
$\det(Dg)(q)
  =\tilde g'(\norm q)\;\bigl(\tilde g(\norm q)/\norm q\bigr)^{d-1}
  =U(q)$.
This shows that $g\rstr_{{\rm int}({\rm supp}(U))}$
is a diffeomorphism onto ${\rm Im}$.

The image measure of $\idty_{{\rm Im}}\,\lambda^d$ under $g^{-1}$
is $\det(Dg)\,\lambda^d$, and $\vep^d \;
{\textstyle \sum_{\ell\in\vep \bZ^d}} 
\;\delta_{\ell} \to \lambda^d$.
So $\lim_{\vep\searrow 0} D_\vep=D_0$.

To prove the statement~\eqref{faltung} on the $C^k$~norm,
we notice that $\vep\mapsto I^{-1}W_{\vep^c}$
is an approximation to the delta distribution~$\delta_0$, and
\beq
I^{-1}{\textstyle \int_{\mathrm{Im}}}\; W_{\vep^c} \bigl(q-g^{-1}(x)\bigr)\,dx
\;=\; U(q)+\cO(\vep^c),
\Leq{U:app}
with corresponding formulae for the derivatives.
%
Then~\eqref{faltung} follows by a direct estimate.
%
\end{proof}

We need to scale natural Hamiltonian systems
\[H\colon P_{\rm ext}\ar\bR\quad , \quad H(t,p,q)=\eh\|p\|^2+V(q)\] 
\nomenclature[APext]{$P_{\mathrm{ext}}$}{extended phasse space}{}{}%
with flows
on extended phase space $P_{\rm ext}:=\bR_t\times T^*\bR^d_q$.
\nomenclature[ATstar]{$T^*$}{cotangential functor}{}{}%
We use three types of scalings which all transform the natural Hamiltonian
and its flow in a simple way:
\begin{enumerate}[1.]
\item \textbf{Motions}
  $M:\bR^d_q\to \bR^d_q$,
  \nomenclature[AM]{$M$}{motion in $\bR_q^d$}{}{}%
  $q\mapsto Oq+v$ with $O\in{\rm SO}(d)$, $v\in\bR^d$
  lift to symplectic maps $M^*$ on $T^*\bR^d_q$,
  \nomenclature[AMstar]{$M^*$}{motion on $T^*\bR_q^d$}{}{}%
  resp.\ $\cM:= {\rm Id}_t\times M^*$ on $P_{\mathrm{ext}}$,
  \nomenclature[AM]{$\cM$}{motion on $P_{\mathrm{ext}}$}{}{}%
  and
  \nomenclature[AMtilde]{$\tilde\cM$}{motion in $C^k(P_{\mathrm{ext}},\bR)$}{}{}%
  \[\tilde \cM: C^k(P_{\rm ext},\bR)\to C^k(P_{\rm ext},\bR)\qtextq,
    \tilde\cM\,  H:= H\circ \cM\text.
  \]
\item \textbf{Spatial scalings ($L^\infty$ dilations)}
  $S_c:P_{\rm ext}\ar P_{\rm ext}$,
  \nomenclature[ASc]{$S_c$}{spatial scaling}{}{}%
  $S_c(t,p,q):=(t/c,\,p,\,c\,q)$, and
  \nomenclature[AStildec]{$S_c$}{spatial scaling on $C^k(P_{\rm ext},\bR)$}{}{}%
  \[\tilde S_c\colon C^k(P_{\rm ext},\bR)\to C^k(P_{\rm ext},\bR)\qtextq,
    \tilde S_c\,H := H\circ S_c\qquad (c>0)\text.
  \] 
\item \textbf{Energy scalings}
  $\cE_e\colon P_{\rm ext}\ar P_{\rm ext}$,
  \nomenclature[AEe]{$\cE_e$}{energy scaling on $P_{\rm ext}$}{}{}%
  $\cE_e(t,p,q):=\bigl(\sqrt{e}\,t,\,p/\sqrt{e},\,q\bigr)$, and
  \nomenclature[AEetilde]{$\tilde\cE_e$}{energy scaling on $C^k(P_{\rm ext},\bR)$}{}{}%
  \[\tilde\cE_e:C^k(P_{\rm ext},\bR)\to C^k(P_{\rm ext},\bR)\quad\text,\quad
    \tilde\cE_e\, H := e\;H\circ \cE_e \qquad (e>0)\text.
  \]
  Notice that, regarding a single site potential~$W$
  as a function on~$P_{\rm ext}$,
  $W_a$ in \eqref{W-scaling}~is a combination of spatial and energy scalings of~$W$.
\end{enumerate}
\begin{proof}[The Case $I > 0$]
$\bullet$
Given $E\in\bR^+$, we choose $\tilde U\in C^k_c\bigl([0,\infty),\bR\bigr)$
with $\tilde U(r)=0$ for all~$r$ smaller than
$2\,{\rm diam}\bigl({\rm supp}(W)\bigr)$ and having one maximum,
with height $\max(\tilde U)\ge 2E$.  
$U=\tilde U\circ\norm{{}\cdot{}}$ then has the property that 
for all $e\in(0,E]$
the energy surface $\{(p,q)\in T^*\bR^d\mid \eh\|p\|^2+U(q)=e\}$ has 
two components.\\
$\bullet$
By \cref{lem:dens} for some $\vep\in(0,\min\{1,I^{((1-c)d)^{-1}\}}]$
we find an approximation of~$U$ of the form
$I^{-1} D_\vep * W_{\vep^c}$ with the same property.
Then, by the combination of spatial and energy scalings, the sum
\begin{equation*}\textstyle
  \sum_{\ell\in\vep \bZ^d \cap\mathrm{Im}}
    W\bigl({}\cdot{}-\tfrac{g^{-1}(\ell)}{\vep^c}\bigr)
    =\tilde\cE_{I\vep^{-(1-c)d}}\circ\tilde S_{\vep^c}
      \bigl(\vep^{-d}D_\vep * W_1\bigr)
\end{equation*}%
of translated single site potentials
has that property for all energies~$e\in(0,I\vep^{-(1-c)d}E]$, too.
Since $\vep^{(1-c)d}\le I$, $E\in(0,I\vep^{-(1-c)d}E]$.\\
$\bullet$
For any realization $\omega\in\Omega$ 
which inside a ball of large radius
is near enough to $\sum_{\ell\in\cL} \delta_{(\ell,1)}$, with $\cL$ as above,
 the energy surface
$\Sigma_{E,\omega}$ consists of at least two connected components.
Thus the motion on $\Sigma_{E,\omega}$ is not ergodic $\beta$--a.s..
\end{proof}

\begin{proof}[The Case $I<0$]
$\bullet$
We first treat the case $d=2$ and then indicate the modifications
needed for larger dimensions.

We use \cref{lem:dens} in order to approximate a 
{\em non--positive} centrally symmetric function~$U$ by
summing scaled translates of the single site function~$W$.
We first assume that the energy is positive, more
specifically $E=1$. We write again $U(q)=\tilde U(\norm q)$ and
choose the profile $\tilde U$ so that there is a circular
periodic orbit $t\mapsto q(t)=r\V{\cos(\omega t)\\\sin(\omega t)}$
of energy~$E$ with $\omega>0$%
\nomenclature[Gomega]{$\omega$}{angular frequency}{}{}
in the potential~$U$.

\begin{figure}
\begin{center}
  \ifpdf
    \includegraphics{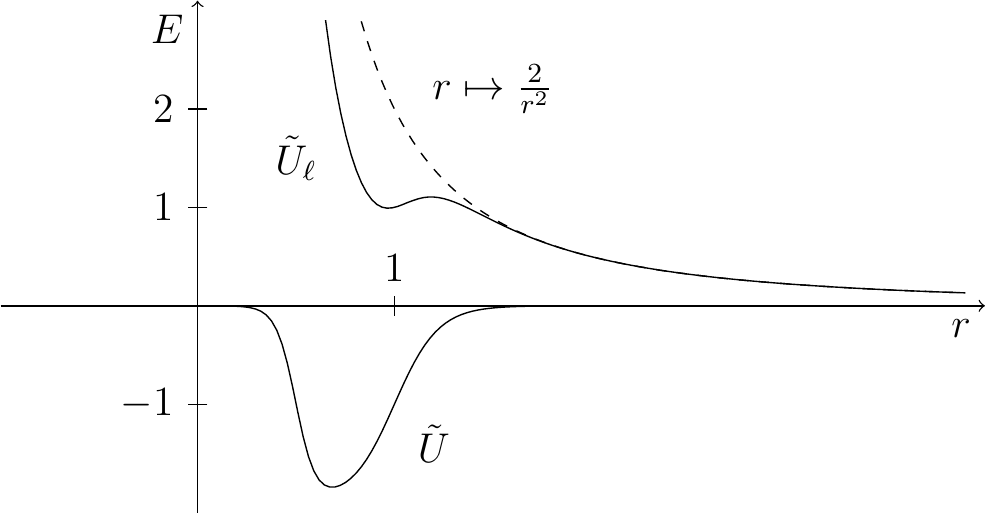}
  \else
    \includegraphics{ueff.ps}
  \fi
\end{center}
\caption{The effective radial potential $\tilde U_\ell$}
\label{fig:radial}
\end{figure}

We thus consider the effective potential 
$\tilde U_\ell(r) := \tilde U(r)+ \frac{\ell^2}{2r^2}$%
\nomenclature[AUtildel]{$\tilde U_\ell(r)$}{effective potential}{}{}
with angular momentum parameter $\ell\in\bR$%
\nomenclature[Al]{$\ell$}{angular momentum parmeter}{}{}.
The conditions for a circular periodic orbit of radius~$r$ are 
$E=\tilde U_\ell(r)$ and $\tilde U'_\ell(r)=0$.
\par
In order to control the stability of the orbit 
(make the linearized flow elliptic and let the frequency
vary with the perturbation), we demand $\tilde U''_\ell(r)>0$
and, say $\tilde U^{(4)}_\ell(r)\neq 0$.

All these conditions can be satisfied by, {\em e.g.}, 
first choosing $r:=1$, $\ell:=2$, and then finding an appropriate 
$\tilde U\in C^\infty_c\bigl([0,\infty),[-1,0]\bigr)$ with $\tilde U(r)=0$
in a neighborhood of zero, $\tilde U(1)=-1$, meeting the assumptions 
on the derivatives and angular frequency $\omega=2$.
Cf.~\cref{fig:radial}.\\
$\bullet$
Then by KAM theory the Hamiltonian flows with potentials $C^k$--near
to $U$ have, too, an elliptic orbit of energy $E=1$, surrounded 
by invariant tori of  positive measure.
See, {\em e.g.}, Arnol'd  \cite{Arn78}, Appendix~8.
According to P\"oschel \cite{Po82} the 
differentiability condition $W\in C^k(\bR^d,\bR)$ 
with $k=3d$ suffices.

Let $\vep\in(0,1)$ be small enough so that the potential 
$\hat{U}:= |I|^{-1} D_\vep * W_{\vep^c}$%
\nomenclature[AUhat]{$\hat U$}{approximation to~$U$}{}{}
(with $D_\vep$ from \cref{lem:dens}) meets that condition for KAM theory. 
Then with the scaling \eqref{W-scaling} of $\hat{U}$, 
$|I|\vep^{-(1-c)d}\hat{U}_{\vep^{-c}}$ is a sum of 
translated single site potentials meeting that condition for 
energy $e:=|I|\vep^{-(1-c)d}$%
\nomenclature[Ae]{$e$}{energy}{}{}. 
This proves the claim for large enough energies~$e$.\\
$\bullet$
In order to solve the problem for an arbitrary smaller energy $E<e$,
we first add to $U$ a function
$\Delta U\in C^\infty_c\bigl(\bR^d,(-\infty,0]\bigr)$,
\nomenclature[AUDelta]{$\Delta U$}{non--positive function}{}{}%
which has a constant value
$\Delta E:=e^{-1}E-1=\abs I^{-1}\vep^{(1-c)d}(E-e)<0$
on the support of~$U$.
Over $\mathrm{supp}(U)$, the flow at energy~$1$ with potential~$U$
equals the flow for energy $\abs I^{-1}\vep^{(1-c)d}E$
with potential $U+\Delta U$.\\
$\bullet$
The case of dimension $d\ge3$ cannot be treated in complete analogy,
since then motion in phase space generated by rotations
from $\mathrm{SO}(d)$ not only leaves the Hamiltonian invariant, but
also can transform a circular orbit to {\em different} circular
orbits. This shows that in the integrable problem periodic
orbits of constant radius have degenerate transverse frequencies.

KAM theory can be applied after lifting this degeneracy.
One way to do this is to construct a potential which 
near the circle 
\[\{q\in\bR^d\mid q_1^2+q_2^2=1, q_3=\ldots= q_d=0\}\]
is of the form $\tilde U(q_1^2+q_2^2)+ \hat U(q_3,\ldots ,q_d)$,
with 
$\hat U(q_3,\ldots,q_d)=\sum_{m=3}^d(\omega_mq_m^2+\tau_mq_m^4)$%
\nomenclature[AUhat]{$\hat U$}{perturbation to lift degeneracies}{}{}
integrable, non--resonant and non--degenerate.
Since such a~$U$ locally is the sum of functions of the
coordinates, $U$~can again be locally 
approximated in a manner similar to \cref{lem:dens}.
\end{proof}
\begin{proof}[The Case $I = 0$]
$\bullet$ Energies $E\le 0$ again lead to bounded motion.\\
$\bullet$ By injectivity of Radon transform (see, \emph{e.g.}, Natterer
\cite{Nat}, Theorem 2.1), there exists a hyperplane $\cH\subseteq\bR^d$%
\nomenclature[AH]{$\cH$}{hyperplane}{}{}
(perpendicular to some $e_1\in S^{d-1}$)
so that $I':=\int_\cH W \,d\lambda^{d-1}\neq 0$. 
By a translation of all the single site potentials $W_j$ 
(which does not change the class of Poisson Hamiltonians)
we can assume that $0\in\cH$.
By our assumption
$I=\int_{\bR^d} W \,d\lambda^{d}= 0$ we can even assume that 
$I'>0$%
\nomenclature[AJ]{$J$}{integral of~$W$ over~$\cH$}{}{}.

We supplement $e_1$ to an orthonormal basis $e_1,\ldots,e_d$ of~$\bR^d$%
\nomenclature[Ae1]{$e_1,\ldots,e_d$}{Basis of $\bR^d$}{}{}
and define the lattice $\cL:={\rm span}_\bZ(e_2,\ldots,e_d)$.
Then for $r>0$ large and the ball $B_r^d\subseteq \bR^d$ of radius~$r$
the linear combination
\beq
\widehat{W}_\vep(q):=
\vep^{d-1}\sum_{\ell\in\vep \cL\cap B_{2r}^d} W(q-\ell)\qquad(q\in\bR^d)
\Leq{Weps}
\nomenclature[AWtildeepsilon]{$\widetilde W_\vep$}{combination of single site potentials}{}{}%
converges, as $\vep\to0$, in $C^k$ sense to
$\widehat{W}_0\in C^\infty_c(\bR^d,\bR)$%
\nomenclature[AWtildezero]{$\widetilde W_0$}{limit of $\widehat W_\vep$}{}{}
with $\widehat{W}_0(q)=\int_{\cH+q_1e_1} W(x)\,dx$
for all~$q$ with $\norm{q-q_1e_1}\le r$ (setting $q_1:=\LA q, e_1\RA$).
So inside that cylinder $\cH\cap B_r\,+\,{\rm span}_\bR(e_1)$,
$\widehat{W}_\vep$ is nearly invariant under translation
perpendicular to~$e_1$, and $\widehat{W}_\vep(q)>0$ if $q_1=0$.

To orbits of regular energy $E\in(0,I')$ 
entering ${\rm supp}\,\widehat{W}_\vep$
with velocity nearly parallel to $e_1$ and position $q$ within the 
cylinder, the potential acts like nearly planar mirror. 

By translating two such potentials by $\pm R e_1$ with $R\gg r$, 
we get a system of two mirrors. Between these mirrors 
trajectories of appropriate energy bounce back and forward 
near the axis ${\rm span}_\bR(e_1)$
for a long time. 

To make the motion bounded for a set of initial
conditions of positive measure, we give the mirrors
inside curvatures stricly smaller than $1/R$. 
This can be done by changing the summand in \eqref{Weps} to
$W(q-\ell-Q(\ell)e_1)$ with an appropriate quadratc form $Q$%
\nomenclature[AQ]{$Q$}{quadratic form}{}{}.
Under this condition the orbit on the axes
becomes linearly elliptic for a potential with perfect 
axial symmetry. Then by KAM theory and appropriate scaling we 
get the result for all positive energies.
\end{proof}

%
\section{No Hyperbolicity for Bounded Potentials}\label{sec6}
%
By far the simplest ergodic flows or maps are the uniformly hyperbolic ones.
Examples include hyperbolic torus automorphisms
and  geodesic flows on
compact manifolds of negative sectional curvature.

Although is known that there exist severe \emph{topological} obstructions
against a flow on a manifold to be Anosov, these do, as shown in \cref{exa}
below, not apply to the motion
in a potential on configuration space $\bR^d$.
However, by \cref{theo} below, 
\emph{geometric} obstructions exist if the potential is bounded.

There exist examples of ergodic motion in smooth bounded potentials 
(which are not uniformly hyperbolic),
see \cite{DL91} by V.\ Donnay and C.\ Liverani.
So our theorem does not exclude ergodicity for concrete smooth potentials
on~$\bR^d$ and some energies.
But it shows that it would be more difficult to prove ergodicity,
and we would not expect ergodicity for open energy intervals.

For $d\ge 2$ we consider the flow $\Phi\colon\bR\times P\to P$ on phase space
$P:= \bR^d_p\times \bR^d_q$ generated by the natural Hamiltonian function
\[H\in C^2(P,\bR)\qtext,H(p,q)=\eh\norm p^2+V(q)\text,\]
where~$V$ and its first and second
derivatives are assumed to be bounded.
This in particular ensures $\Phi\in C^2(\bR\times P, P)$.

We call the flow $\Phi_E\colon\bR\times \Sigma_E\to\Sigma_E$ restricted 
to an energy surface $\Sigma_E:= H^{-1}(E)$  \emph{Anosov} if 
there is a $d\Phi_E$--invariant splitting
\nomenclature[Ad]{$d$}{exterior derivative}{}{}%
\[T_x\Sigma_E = {\rm span} \bigl(X_H(x)\bigr) \oplus E^u(x)\oplus E^s(x)\]
\nomenclature[AT]{$T$}{tangential functor}{}{}%
\nomenclature[AXH]{$X_H$}{hamiltonian vectorfield}{}{}%
\nomenclature[AEu]{$E^u$}{strong unstable bundle}{}{}%
\nomenclature[AEs]{$E^s$}{strong stable bundle}{}{}%
into a one-dimensional bundle spanned by the Hamiltonian vector field $X_H$,
and the strong (un)stable bundles $E^{u/s}$, along which $d\Phi_E$
is exponentially contracting in backward resp.\ forward time.
Even if $\Sigma_E$ is not compact, this is unambiguously defined 
by existence of $C\ge1$ and $\lambda>0$ with
\begin{equation*}
  \norm[\Phi_E(t,x)]{(d\Phi_E^t)_x (v)}
    \le C\exp(-\lambda t)\norm[x]v
      \qquad\bigl(x\in\Sigma_E,v\in E^s(x),t\in[0,\infty)\bigr)\text,
\end{equation*}
\nomenclature[Glambda]{$\lambda$}{contraction exponent}{}{}%
(and analogously for $E^u$) if we take translation--invariant 
norms $\norm[x]{{}\cdot{}}$
on the tangent spaces~$T_xP$ in the bundle $TP$.

\begin{example}\label{exa}
  For the potential $V(q)=-\eh\norm q^2$ ($q\in\bR^d$), the flow equals
  \begin{equation*}
    \V{p(t)\\q(t)}
      =\V{\idty\cosh(t)&\idty\sinh(t)\\
	  \idty\sinh(t)&\idty\cosh(t)}\V{p(0)\\q(0)}
    \qquad(t\in\bR)\text.
  \end{equation*}
  So we can take $\lambda=C=1$, and 
  $E^{u/s}(x)=\{\V{a\\\pm a}\mid a\in\bR^d\}$.
  We see that on~$P$ there is no \emph{topological}
  obstruction against the motion generated by~$H$
  to be Anosov (but we remark that here $\dim\bigl(E^{u/s}(x)\bigr)=d$).
\end{example}

The following statement follows from specializing a theorem in \cite{PP94}:
\begin{theorem}[G.\ P.\ and M.\ Paternain]
{\em For no value $E< \sup_q V(q)$ the flow $\Phi_E$ on $\Sigma_E$ is
non--wandering and Anosov.}
\end{theorem}
\begin{proof}
\begin{itemize}[$\bullet$]
\item
To be Anosov, $E$~must be a regular value of~$H$ or -- equivalently -- of~$V$,
since otherwise there are points on $\Sigma_E$,
where the Hamiltonian vector field $X_H$ vanishes.
So one assumes that~$E$ is a regular value, that $\Phi_E$ is
non--wandering and Anosov, and derives a contradiction.
\item
Regularity of~$E$ is one of the assumptions of Theorem~3
in \cite{PP94}.
The condition of existence of a
$\Phi_E$--invariant lagrangian subbundle~$E$ of $T\Sigma_E$ is met,
too, by the centre stable bundle, with 
$E(x):= {\rm span} \bigl(X_H(x)\bigr) \oplus E^s(x)\subseteq T_x\Sigma_E$.
\item
As a conclusion Theorem 3 in \cite{PP94} states that $E$ trivially intersects
the vertical bundle  ${\rm Vert}$, given for $x\in P\cong T^*\bR^d_q$ by
${\rm Vert}(x)={\rm ker}(d\pi_x)$, with the projection
$\pi\colon T^*\bR^d_q\to\bR^d_q$, $(p,q)\mapsto q$.

This implies that $E\ge \sup_q V(q)$.
For otherwise there is a point $x=(p,q)\in \Sigma_E$ with $V(q)=E$.
By regularity of the value~$E$ then $0\neq X_H(x)\in {\rm Vert}(x)$.
\end{itemize}
\end{proof}
\begin{remark}[The Non-Wandering Condition]
\begin{enumerate}[1.]
\item
If one wants to show ergodicity using the Anosov property,
then the non--wandering condition is somewhat natural:

The Liouville measure $\mu_E$ on $\Sigma_E$
is invariant under the flow $\Phi_E$.
For a regular value~$E$ of~$H$, $\mu_E$ is absolutely continuous
w.r.t.\ the riemannian measure.
The flow is called \emph{ergodic} (in the sense of Aaronson, see \cite{Aar97}) if every 
$\Phi_E$--invariant measurable subset~$A$ of $\Sigma_E$ is of
measure zero or the complement of a measure zero set.

Under the assumption of ergodicity, the non--wandering set of 
$\Phi_E$ equals~$\Sigma_E$.
For assume that $x\in \Sigma_E$ is wandering.
Then by definition there
is an open neighbourhood $U\subseteq \Sigma_E$%
\nomenclature[AU]{$U$}{neighbourhood of $\Sigma_E$}{}{}
of~$x$ and $T>0$, so that $\Phi_E(t, U) \cap U = \es$ if $\abs t\ge T$.
Since $\dim(\Sigma_E)>1$, there is a neighbourhood $W\subseteq U$%
\nomenclature[AW]{$W$}{neighbourhood of~$U$}{}{}
of~$x$ so that $\mu_E\bigl(\bigcup_{t\in [-T,T]} \Phi_E(t, W)\bigr)
  < \mu_E\bigl(U\bigr)$.
So both the $\Phi_E$--invariant set $\bigcup_{t\in \bR} \Phi_E(t,W)$
and its complement have positive measures, contradicting ergodicity.
\item
Besides that, there are many alternatives to the 
non--wandering condition in the above theorem. 

One choice is to assume that the boundary $V^{-1}(E)\subseteq \bR^d_q$ 
of Hill's 
region contains a compact component which is not diffeomorphic to
$S^{d-1}$, or more than one compact component.
Then there exists a closed ({\em brake}) orbit with positive Maslov class, 
contradicting the existence of a section of the lagrangian bundle
(see Theorem 2 of \cite{Kn90} and Section 6 of \cite{KK08}).

Another such alternative is the assumption that there exists an $e<E$
such that for every $r>0$ there exists a ball 
$B_r(Q)\subseteq \bR^d_q$ with $V\rstr_{B_r(Q)}\le e$. Then the 
proof of \cref{theo} below can be adapted. 
\end{enumerate}
\end{remark}
For large energies we do not need the non--wandering condition.
\begin{theorem}\label{theo}
For no $E> \sup_q V(q)$ the flow $\Phi_E$ on $\Sigma_E$ is Anosov.
\end{theorem}
\begin{proof}
\begin{itemize}[$\bullet$]
\item
For $E>\sup_q V(q)$ we use riemannian geometry.
The metric~$g$ of a riemannian manifold~$(M^d,g)$ 
defines a connection and thus a canonical decomposition of~$T(TM)$
into a horizontal and a vertical subspace:
\[T_xTM=T_{x,h}TM\oplus T_{x,v}TM\qquad(x\in TM)\text.\]
Both $T_{x,h}TM$ and $T_{x,v}TM$ are canonically isomorphic to the
$d$-dimensional vector space $T_q M$ (for $x\in T_q M$).

For a lagrangian subspace $\lambda\subseteq T_x TM$ which 
is transversal to the vertical subspace, there exists a symmetric operator 
\begin{equation}\label{OpS}
  S\colon T_{x,h}TM\to T_{x,v}TM
\end{equation}
\nomenclature[AS]{$S$}{symmetric operator $T_{x,h}TM\to T_{x,v}TM$}{}{}%
such that the vertical and horizontal component of any vector
$w=w_h+w_v\in \lambda$ obey the relation (see, \emph{e.g.}, 
Klingenberg \cite[3.2.16~Proposition]{Kli95})
\begin{equation*}
  w_v=Sw_h\text.
\end{equation*}
The covariant derivative $\nabla Y(t)$
\nomenclature[GnablaY]{$\nabla$}{covariant derivative}{}{}%
of a stable Jacobi field~$Y(t)$ along
a geodesic trajectory equals $S(t)Y(t)$.
Hence the operator~$S$ satisfies the Riccati equation 
\begin{equation}\label{Ric}
  S^2=-\nabla S-R_X\text.
\end{equation}
\nomenclature[ARX]{$R_X$}{curvature operator along~$X$}{}{}%
along the geodesic.
In our case we use on $M:=\bR^d_q$ the Jacobi-Maupertuis metric~$g$,
with $g(q):=\bigl(E-V(q)\bigr)\cdot g_{\textrm{Euclid}}(q)$.
Up to a reparametrisation of time~$t$, 
the geodesics in this metric coincide with the projection
of the solutions $t\mapsto \Phi_E(t,x)$
of our Hamiltonian equation to configuration space~$M$.
Since $E-V(q)$ is bounded from below and above by positive constants,
the derivative of time reparametrisation is bounded below and above, too.
We denote the geodesic flow by
\[\Psi:\,\bR\times T_1M\to T_1M\text.\]
\nomenclature[GPsi]{$\Psi$}{geodesic flow}{}{}%
By Theorem~3 of \cite{PP94} we can write the lagrangian subbundle~$E$
as the graph of a symmetric operator valued function of the form~\eqref{OpS}.
 
We integrate the trace of~\eqref{Ric}
over the unit tangent bundle $T_1B_r$
of the ball $B_r=B_r(0)\subseteq \bR^d$ of radius~$r$.
\item
The integral of the covariant derivative is of order 
\begin{equation}
\int_{T_1 B_r} \operatorname{trace}(\nabla S)\, dm\, do = \cO(r^{d-1})\text,
\label{surface}
\end{equation} 
where we denote by $dm(q)=\sqrt{\det g(q)}dq_1\wedge\ldots\wedge dq_d$%
\nomenclature[Am]{$m$}{riemannian measure}{}{}
the measure on~$M$ and by $do$%
\nomenclature[Ao]{$o$}{measure on $S^{d-1}$}{}{}
the measure on the unit sphere ($\int_{S^{d-1}}do=\operatorname{vol}(S^{d-1})$).
We show \eqref{surface} by reducing it to a term scaling with
the volume of the boundary $\pa B_r$ of the ball.
To this end we decompose the region $T_1 B_r$ of the energy surface into
\[T_1 B_r\; =\;\cS\;\dot{\cup}\;\cB\;\dot{\cup}\;\cT\text,\]
using the maximal time interval $I(x)$%
\nomenclature[AI]{$I$}{time interval}{}{}
containing~$0$ for which the geodesic flow line through~$x$
stays within $T_1B_r$:
\begin{itemize}[-]
  \item The \emph{scattering set}\,\footnote{The names should not be taken too 
    serious, since, \emph{e.g.}, the intersection of a $\Psi$--orbit with $\cS$
    can consist of several components.}~%
    $\cS:=\bigl\{x\in T_1 B_r\mid I(x)=[T^-(x),T^+(x)]\;\bigr\}$,
    \nomenclature[AS]{$\cS$}{scattering set}{}{}%
  \item the \emph{bounded set}~%
    $\cB:=\bigl\{x\in T_1B_r\mid I(x)=\bR\bigr\}$%
    \nomenclature[AB]{$\cB$}{bounded set}{}{}
    and
  \item the \emph{trapped set}~%
    $\cT:=T_1 B_r\setminus(\cB\,\cup\,\cS)$%
    \nomenclature[AT]{$\cT$}{trapped set}{}{}.
\end{itemize}
All three sets are measurable.
\begin{itemize}[-]
  \item
    The trapped set $\cT$ consists of wandering points
    and thus is of measure zero.
  \item
    The bounded set $\cB$ is $\Psi$--invariant.
    So we can use the relation
  \beq
    \int_0^T\operatorname{trace}\bigl(\nabla S\circ\Psi(t,x)\bigr)\,dt = 
    \operatorname{trace}\bigl(S\circ\Psi(T,x)-S(x)\bigr)
  \Leq{Hauptsatz}
  to show 
  \begin{align*}
    T\int_{\cB}
    \operatorname{trace}(\nabla S)\,dm\,do&
      =\int_{\cB}\int_0^T
	\operatorname{trace}\bigl(\nabla S\circ\Psi(t,x)\bigr)\,dt\,dm\,do\\& 
      =\int_{\cB}\operatorname{trace}\bigl(S\circ\Psi(T,x)-S(x)\bigr)\,dm\,do
      =0\text.
  \end{align*}
  \item
  So the only contribution to~\eqref{surface}
  comes from the scattering set~$\cS$.
  Every $x\in\cS$ can be uniquely written as
  $x=\Psi(t,y)$ with $t\in[0,T^+(y)]$ and $T^-(y)=0$.
  Conversely, for the points in 
  $\cV:=\{y\in \cS\mid T^-(y)=0\}$ all
  $\Psi(t,y)$ with $t\in[0,T^+(y)]$ are in $\cS$.
  So we rewrite the integral:
  \[\int_\cS\operatorname{trace}(\nabla S)\,dm\,do=\int_\cV\int_0^{T^+(y)} 
  \hspace*{-4mm}\operatorname{trace}\bigl(\nabla S\circ\Psi(t,y)\bigr) J(y)
  \,dt\, dy\text.\]
  Since the Jacobian $J\colon\cV\to\bR^+$ is bounded above by~$1$,
  and $\operatorname{trace}(S)$ is bounded on~$T_1M$,
  we obtain \eqref{surface}, reusing \eqref{Hauptsatz}.
\end{itemize}
\item
For the second term on the right hand side of \eqref{Ric},
\[\int_{T_1 B_r} \operatorname{trace}(R_X)\, dm\, do = 
\frac{\operatorname{vol} (S^{d-1})}{d}\int_{B_r} {\cal R}(q)\, dm\text,\]
\nomenclature[AR]{$\mathcal R$}{scalar curvature}{}{}%
where ${\cal R}(q)$ denotes the scalar curvature.
If the motion takes place on a two-dimensional plane $M=\bR^2_q$, then
$\int_{B_r}{\cal R}(q)\,dm=\cO(r)$ as a consequence of the Gauss-Bonnet formula.
For dimension $d\ge3$, that equality is wrong in general.
But in our case the Jacobi metric is conformally flat.
Defining the positive function $u\colon M\to\bR^+$ by
$u(q):=\bigl(E-V(q)\bigr)^{(d-2)/4}$, the measure $dm$ on~$M$
equals
$dm = u^{\frac{2d}{d-2}}dq_1\wedge\ldots\wedge dq_d$.
The scalar curvature equals
\[{\cal R}=\frac{1-d}{(E-V)^3} \l[ (V-E)\Delta V + \frac{d-6}{4}
(\nabla V)^2 \ri] = 4\frac{1-d}{d-2}u^{-\frac{d+2}{d-2}}\Delta u\]
(with the euclidean Laplacian $\Delta=\sum_{k=1}^d \frac{\pa^2}{\pa q_k^2}$).
\nomenclature[GDelta]{$\Delta$}{euclidian Laplacian}{}{}%
Therefore
\begin{align}
  \lefteqn{\int_{B_r}{\cal R}\,dm
    =-4\frac{d-1}{d-2}\int_{B_r}u^{-\frac{d+2}{d-2}}
  (\Delta u) u^{\frac{2d}{d-2}}dq_1\wedge\ldots\wedge dq_d}\nonumber\\&
    =-4\frac{d-1}{d-2}\int_{B_r} u(\Delta u) 
			       dq_1\wedge\ldots \wedge dq_d\nonumber\\&
    =+4\frac{d-1}{d-2}\int_{B_r} (\nabla u)(\nabla u) 
			       dq_1\wedge\ldots \wedge dq_d + f(r) 
    \ge  f(r)\text.\label{larger}
\end{align}
The surface integral is of order $f(r)=\cO(r^{d-1})$.
Concerning the left hand side of \eqref{Ric},
the integral of $\operatorname{trace}(S^2)$ is positive.
For uniform hyperbolicity, one would need
$\limsup_{r\to\infty}r^{-d}\int_{B_r}\operatorname{trace}(S^2)(q)\,dm>0$, 
which is impossible, since the corresponding $\limsup$ of the right hand side is
nonpositive.
So \eqref{larger} is compatible with \eqref{Ric} only if
the flow is not Anosov.
\end{itemize}
\end{proof}
%
\section{Random Coulombic Potentials}\label{sec9}
%
As it does not seem to be so simple to find smooth random potentials that 
lead to ergodic motion, we now study the example of coulombic potentials,
see \cref{fig:random} for a numerical realization.
\begin{figure}
\begin{center}
  \ifpdf
    \includegraphics[width=5cm,clip]{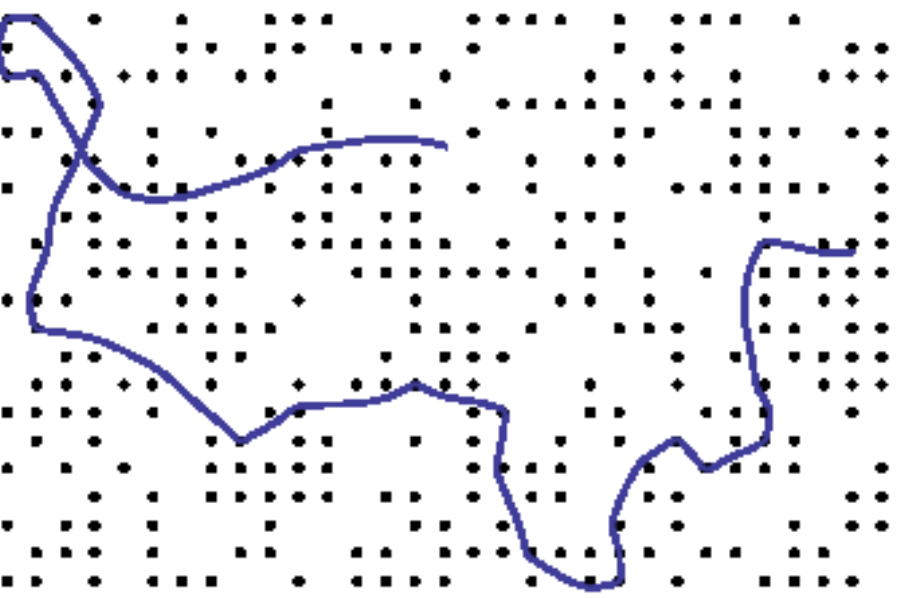}
  \else
    \includegraphics[width=5cm,clip]{simulation3.ps}
  \fi
\end{center}
\caption{Motion in the configuration space of a random coulombic potential}
\label{fig:random}
\end{figure}
We restrict ourselves to dimension $d=2$.
Here for $j\in J$ the single site potentials
\[\tilde W_j\in C^\eta(\bR^2\setminus\{s_j\},\bR)\]
\nomenclature[AWtildej]{$\tilde W_j$}{single site potential with singularity}{}{}%
(with $\eta\in\bN\cup\{\infty\},\eta\ge2$)
diverge at the position $s_j\in \cD$ in the fundamental domain
$\cD:=\{x_1\ell_1+x_2\ell_2\mid 0\leq x_i<1\}$
of the lattice $\cL=\Span_\bZ(\ell_1,\ell_2)$.

Our assumptions are:
\begin{compactenum}
  \item\label{shortrange}
    \emph{The decay at infinity} is short range,
    that is for $\alpha\in\bN_0^2,\ \abs\alpha\leq\eta$
    \[\pa^\alpha\tilde W_j(q)=\cO(\norm q^{-2-\vep})\qquad(\norm q\rightarrow\infty)\text.\]
  \item\label{coulomb}
    The \emph{local singularity} at $s_j\in\bR^2\cong\bC$ is controlled by
    \[f_j\colon\bC^*:=\bC\setminus\{0\}\to\bR \qtextq{,}
    f_j(z):=\abs z^2\tilde W_j(z^2+s_j)\]
    and we assume that for all multi--indices $\alpha\in\bN_0^2$,
    $\abs\alpha\leq\eta$, $\pa^\alpha f_j$
    can be continuously extended to zero, with $f_j(0)<0$.
    We allow for an additional single site potential $\tilde W_0=0$.
\end{compactenum}
\begin{example}
  \label{ex:Yukawa}
  The \emph{Yukawa Potential} with parameters $c_j,\mu_j>0$
  is defined via
  \[\tilde W_j(q)=-c_j\frac{\exp(-\mu_j\norm{q-s_j})}{\norm{q-s_j}}\text.\]
\end{example}
\begin{example}
  \label{ex:finiteRange}
  \emph{Finite range potentials}
  are given by
  \[\tilde W_j(q):=-\frac{g_j(\norm{q-s_j})}{\norm{q-s_j}}\text,\]
  with $g_j\in C_c^\eta(\bR,\bR),\ g_j(0)>0$.
  \nomenclature[Agj]{$g_j$}{radial shape of finite range potential}{}{}%
\end{example}
The random potential is determined by the probability space
$(\Omega,\cB(\Omega),\beta)$ with $\Omega=\cL^J$.
The probability measure~$\beta$ is assumed to be $\vartheta$--invariant,
see equation~\eqref{eq:vartheta}, 
and to give probability $\beta\{\omega_0\}=0$
to the configuration $\omega_0\in\Omega$ with $\omega_0(\ell)=0$, $\ell\in\cL$.
No $\cL$--ergodicity of $\beta$ is assumed here.

For what follows we fix $\omega\in\Omega$.
The \emph{punctured configuration space}
$\tilde M_\omega:=\bC\setminus\cS_\omega$
\nomenclature[AMtildeomega]{$\tilde M_\omega$}{punctured configuration space}{}{}%
now depends on the \emph{singularity set}
$\cS_\omega:=\{s_{\omega(\ell)}+\ell\mid\ell\in\cL,\omega(\ell)\ne0\}$.
\nomenclature[ASomega]{$\cS_\omega$}{singularity set}{}{}%
The former supports the random potential
\begin{equation*}
  \tilde V_\omega\colon\tilde M_\omega\to\bR\qtextq,
  \tilde V_\omega(q):=\sum_{\ell\in\cL}\tilde W_{\omega(\ell)}(q-\ell)\text.
\end{equation*}
\nomenclature[AVtildeomega]{$\tilde V_\omega$}{coulombic potential}{}{}%
Assumption~\ref{shortrange}.\ guarantees the convergence of $\tilde V_\omega$
and its derivatives, 
and Assumption~\ref{coulomb}.\ implies that the Coulombic singularities
are attractive and
\begin{equation*}
  \tilde V_{\omega,\max}:=\sup\tilde V_\omega(\tilde M_\omega)<\infty\text.
\end{equation*}
\nomenclature[AVtildeomegamax]{$\tilde V_{\omega,\max}$}{supremum of $\tV_\omega$}{}{}%
The Hamiltonian flow $\tilde\Phi_\omega\colon\tilde U_\omega\to\tilde P_\omega$
\nomenclature[GPhitildeomega]{$\tilde Phi_\omega$}{incomplete hamiltonian flow}{}{}%
on $\tilde P_\omega:=T^*\tilde M_\omega=\bR^2\times\tilde M_\omega$,
generated by the Hamiltonian function
$\tilde H_\omega\colon\tilde P\to\bR$, $(p,q)=\frac12\norm p^2+\tilde V_\omega(q)$,
\nomenclature[GHtildeomega]{$\tilde H_\omega$}{coulombic hamiltonian function}{}{}%
is now, due to the Coulombic singularities, incomplete and only defined
on a maximal open subset $\tilde U_\omega\subseteq\bR\times\tilde P_\omega$.
\nomenclature[AUtildeomega]{$\tilde U_\omega$}{maximal domain of $\tilde Phi$}{}{}%

It is known that the flow can be continuously
regularised by reflecting collision orbits at their singularity.
This is possible, see Prop.\ 2.3 of \cite{KK92} or
Thm.~11.23 of \cite{Kn11},
by smoothly extending the incomplete Hamiltonian system 
\[\big(\tilde P_\omega,dq\wedge dp,\tilde H_\omega\big)
  \quad\text{ to a Hamiltonian system }\quad 
  \big(P_\omega, \sigma_\omega, H_\omega\big)\text,\]
\nomenclature[APomegasigmaomega]{$(P_\omega,\sigma_\omega)$}{symplectic manifold}{}{}%
\nomenclature[Gsigmaomega]{$\sigma_\omega$}{symplectic form}{}{}%
with $H_\omega\colon P_\omega\to \bR$
generating a complete smooth Hamiltonian flow
on the symplectic manifold $(P_\omega, \sigma_\omega)$.
In fact, similarly to the construction in \cref{sec2},
these data extend continuously to a triple $(P,\sigma,H)$, with
extended phase space 
$P:= \bigcup_{\omega\in\Omega}\ \{\omega\}\times P_\omega$.

For $d=2$ and energies $E>V_{\max}$ regularization can also be performed, 
as in \cite{Kna87,KK92}, with the help of the twofold covering
\begin{equation}\label{eq:Momega}
  \pi_\omega\colon\Mo:= \big\{(q,Q)\in\bC^2\mid f_\omega(q)=Q^2\big\}
   \; \longrightarrow \;\bC
\end{equation}
\nomenclature[Gpiomega]{$\pi_\omega$}{twofold covering $\Mo\to\bC$}{}{}%
\nomenclature[AMomega]{$\Mo$}{branched covering surface}{}{}%
with branch points in the singularity set
$\cS_\omega$, where $f_\omega\colon\bC\to\bC$
\nomenclature[AMomega]{$f_\omega$}{entire function}{}{}%
is an entire function with simple zeroes in and only in~$\cS_\omega$:
$f_\omega^{-1}\{0\}=\cS_\omega$.
By a Weierstrass product construction
we can choose $(\omega,q)\mapsto f_\omega(q)$
as a continuous function.
The lift $\mathbf{\tilde g}_{\omega,E}:=(\pi_\omega)^*\tilde g_{\omega,E}$
\nomenclature[AgtildeomegaE]{$\mathbf{\hat g}_{\omega,E}$}{Jacobi metric on $\Mo$ without branch points}{}{}%
of the Jacobi-Maupertius metric
\begin{equation*}
  \tilde g_{\omega,E}(q):=\bigl(1-E^{-1}\tilde V_\omega(q)\bigr)g_{\Euclid}(q)
  \qquad(q\in\tM_\omega)
\end{equation*}
can be continued to all of $\Mo$ by taking limits:
\begin{equation*}
  \mathbf g_{\omega,E}(q,0):=
    \lim\limits_{\Mo\ni(q',Q)\to(q,0)}\mathbf{\hat g}_{\omega,E}(q',Q)\qquad
      (\pi_\omega(q)\in\cS_\omega)\text,
\end{equation*}
\nomenclature[AgomegaE]{$\mathbf g_{\omega,E}$}{Jacobi metric on $\Mo$}{}{}%
see \cite{Kna87,KK92}.
This gives a smooth and complete riemannian metric on~$\Mo$,
and its geodesics are, up to a reparametrisation of time
and modulo~$\pi_\omega$, trajectories of the Hamiltonian flow.

This geometric regularisation allows to take full advantage of riemannian geometry.
For two-dimensional compact riemannian manifolds of negative curvature
ergodicity of the geodesic flow was established by \cite{Hop41}.
The gaussian curvature of $(\tM_\omega,\hat g_{\omega,E})$
at $q\in\tM_\omega$ is given by
\begin{equation}\label{eq:curvature}
  \begin{split}
    \tilde K_{\omega,E}(q)&
      =\frac{\bigl(E-\tilde V_\omega(q)\bigr)\Delta\tilde V_\omega(q)
	      +\bigl(\nabla\tilde V_\omega(q)\bigr)^2}
	    {2\bigl(E-\tilde V_\omega(q)\bigr)^3}
    \\&
      =\frac{\bigl(1-E^{-1}\tilde V_\omega(q)\bigr)E^{-1}
	      \Delta\tilde V_\omega(q)
	      +\bigl(E^{-1}\nabla\tilde V_\omega(q)\bigr)^2}
	    {2E\cdot\bigl(1-E^{-1}\tilde V_\omega(q)\bigr)^3}
  \text.
  \end{split}
\end{equation}
\nomenclature[AKtildeomegaE]{$\tilde K_{\omega,E}$}{gaussian curvature}{}{}%
Since $\pi_\omega^{-1}(\tilde M_\omega)$ is a local isometry,
the curvature of all non--branch points of~$\Mo$
is determined by \cref{eq:curvature}.
In \cref{ex:Yukawa} 
it turns out that for high enough energy
$E>E_{\text{th}}>\tilde V_{\omega,\max}$ the curvature is non--positive.
This motivates the following definition.
\begin{defi}
  The pair $(\tilde V_\omega,E)$ is of
  \emph{non--positive (strictly negative) curvature},
  if $(\Mo,\mathbf g_{\omega,E})$
  has non--positive (strictly negative) curvature.
\end{defi}
A direct consequence of equation~\eqref{eq:curvature} is, that,
given a pair $(\tilde V_\omega,E)$ of non--positive curvature,
all $(\tilde V_\omega,E')$ with $E'\ge E$
are of non--positive curvature, too.
As $E\nearrow\infty$ the curvature concentrates in the singularities.
In \cref{ex:finiteRange} 
non--positive curvature can be achieved
for all energies above a threshold with suitable choices for $g_j$, \emph{e.g.}
\begin{equation*}
  g_j(r):=
  \begin{cases}
    -c_j\cos^{\eta+1}(\lambda_jr)
      &r<\frac\pi{2\lambda}\\
    0 &r\ge\frac\pi{2\lambda}
  \end{cases}
\end{equation*}
with $c_j,\lambda_j>0$, see \cite{Kna87}.  From now on we will assume non--positive curvature.
This assumption and the following lemma explain why we allow only $W_0=0$
as smooth single site potential.
\begin{lemma}
  If a single site potential $\tilde W\colon\bC\to\bR$
  without singularity and an energy~$E$
  are of nonpositive curvature, then $\tilde W=\tilde W_0\equiv0$.
\end{lemma}
\begin{proof}
  From \eqref{as:sisi} it follows that the curvature of one single site potential is integrable.
  By the theorem of Gauss--Bonnet
  the integral of the curvature over~$\bC$ vanishes.
  Since the curvature is nonpositive, it has to vanish, too,
  and so does the single site potential.
\end{proof}

Our goal is
\begin{theorem}\label{thm:toptrans}
  If $\,(\tilde V_\omega,E)$ is of non--positive curvature, then 
  the geodesic flow
  $\Phi_\omega\colon\bR\times T_1\Mo\to T_1\Mo$
  \nomenclature[AT1]{$T_1$}{unit tangent bundle functor}{}{}%
  is \emph{topologically transitive} $\beta$--almost surely.
\end{theorem}

The strategy will be as follows.
In \cref{prop:periodicdense} we show that
the periodic orbits of~$\Phi_\omega$ are dense in~$T_1\Mo$.
Note, that the same is true for Anosov diffeomorphisms on compact manifolds,
cf.~\cite[3.8]{Bow75}.
Then, in the proof of \cref{thm:toptrans} from page~\pageref{proof-toptrans} on,
in order to connect two open subsets of $T_1\Mo$,
we will connect two periodic orbits with an intertwining orbit.

\newcommand{\vphi}{\varphi}%
\newcommand{\cone}[3]{\vartriangle_{#1}^{(#2)}\!(#3)}%
\newcommand{\sector}[3]{\blacktriangle_{#1}^{(#2)}(#3)}%
We start with a useful lemma, which singles out the set of full measure
on which we establish topological transitivity.
We denote the Euclidean cone intersected with the lattice~$\cL$
\begin{equation*}
  \cone q\vphi x:=\{\ell\in\cL\mid\sphericalangle(x,\ell-q)<\vphi\}\text,
\end{equation*}
\nomenclature[Gdelta]{$\cone q\vphi x$}{euclidian cone}{}{}%
$q\in\bC$, $\vphi\in[0,\pi]$, $x\in T_{1,q}\bC=S^1$.

\begin{lemma}\label{lemma:nofreecones}
  The $\vartheta$--invariance of $\beta$ and $\beta\{\omega_0\}=0$
  imply that there are $\beta$-almost surely no Euclidean cones
  without nonvanishing single site potential.
  More precisely:
  \begin{equation*}
    \beta\bigl\{\omega\in\Omega\bigm|
      \exists(q,\vphi,x)\in\bR^2\times(0,\pi]\times
      S^{1}\colon
      \omega(\cone q\vphi x)=\{0\}\bigr\}=0\text.
  \end{equation*}
\end{lemma}
\begin{proof}
  We study the sets
  \begin{equation*}
    A_{q,\vphi,x}:=\bigl\{\omega\in\Omega\bigm|
      \omega\bigl(\cone q\vphi x\bigr) = \{0\} \bigr\}\text,
  \end{equation*}
  $(q,\vphi,x)\in\bR^{2}\times(0,\pi]\times
  S^{1}$, first.
  For $(q,\vphi,\ell)\in\bR^2\times(0,\pi]\times\cL$ we have
  \begin{equation*}
    \beta\bigl(A_{q,\vphi,\frac\ell{\norm\ell}}\bigr)
      =\lim_{n\to\infty}\beta\bigl(A_{q-n\ell,\vphi,\frac\ell{\norm\ell}}\bigr)
      =\beta\Bigl(\Isect\nolimits_{n\in\bN}A_{q-n\ell,\vphi,\frac\ell{\norm\ell}}\Bigr)
      =\beta\{\omega_0\}
      =0\text.
  \end{equation*}
  Note now that the set in question can be written as
  the denumerable union
  \begin{equation*}
    \Union\nolimits_{(q,\vphi,\ell)\in\bQ^2\times\bQ_+\times\cL}
      A_{q,\vphi,\frac\ell{\norm\ell}}\text.
  \end{equation*}
\end{proof}

We denote the tangent bundle of $\Mo$ with $\tau_{\Mo}\colon T\Mo\to\Mo$
\nomenclature[GtauMomega]{$\tau_{\Mo}$}{tangential projection}{}{}%
and the natural length metric on $(\Mo,\mathbf g_{\omega,E})$ with $d_{\Mo}$.
\nomenclature[AdMomega]{$d_{\Mo}$}{natural length metric on $\Mo$}{}{}%
\begin{lemma}\label{lemma:singsappear}
  For $\beta$--almost all $\omega\in\Omega$ the following holds.
  Given a point $q\in\Mo$ and a nonempty open subset
  $U\subseteq T_{1,q}\Mo:=T_1\Mo\isect T_q\Mo$ of the sphere bundle,
  there exists an initial direction $v\in U$
  such that the corresponding geodesic hits a branch point at $t>0$, i.e.
  \begin{equation*}
    \exp_{\Mo}(tv)\in\pi_\omega^{-1}(\cS_\omega)\text.
  \end{equation*}
\end{lemma}
\begin{proof}
  W.l.o.g.\ we assume $U\ne T_{1,q}\Mo$.
  For every $\omega\in\Omega$ such that there is an open subset
  $U\subseteq T_{1,q}\Mo$ with no branch points in
  \begin{equation*}
    \blacktriangle:=\exp_{\Mo}(\bR_+U)\;\subseteq\;\Mo
  \end{equation*}
  \nomenclature[Gdelta]{$\blacktriangle$}{Jacobian cone}{}{}%
  we will establish the existence of a Euclidean cone
  $\vartriangle\subseteq\cL$ with $\omega(\vartriangle)=\{0\}$.
  Then \cref{lemma:nofreecones} applies and gives the desired result.
  \par
  Fix $v\in U$.
  Since the curvature of $(\Mo,\mathbf g_{\omega,E})$ is nonpositive,
  we find a constant $\rho_0>0$ such that
  \begin{equation}\label{eq:bigdistance}
    d_{\Mo}\bigl(\exp_{\Mo}(tv),\,\Mo\setminus\blacktriangle\bigr)
      \;\ge\;2\rho_0t
  \end{equation}
  for all $t\ge0$, see \cite{BBI01}.
  Note that, due to the absence of branch points in~$\blacktriangle$,
  $\pi_\omega|_{\blacktriangle}$ is a homeomorphism onto its image,
  and since the conformal factor is bounded by
  $h:=1-E^{-1}\tilde V_{\omega,\max}$, we know for all $q,q'\in\Mo$
  \begin{equation*}
    d_{\Mo}(q,q')
      \;\le\; h \cdot\abs{\pi_\omega(q)-\pi_\omega(q')}\text.
  \end{equation*}
  This implies
  \begin{equation*}
    \blacktriangle'
      \;:=\;\Union_{t\ge0}B_{h^{-1}\rho_0t}
	\bigl(\pi_\omega\circ\exp_{\Mo}(tv)\bigr)
      \;\subseteq\;\pi_\omega(\blacktriangle)
      \;\subseteq\;\bC\text,
  \end{equation*}
  and in $\blacktriangle'$ we search our Euclidean cone~$\vartriangle$.
  All we have to show to this end is that the ``axis'' of $\blacktriangle'$
  \begin{equation*}
    \bR_+\ni t\;\longmapsto\;\bigl(q_\omega(t),p_\omega(t)\bigr)
      :=T\pi_\omega\circ\Phi_\omega(t,v)
  \end{equation*}
  converges to a definite direction:
   $ \lim_{t_0\to\infty}
      \sup_{t\ge t_0} \abs{p_\omega(t+t_0,v)-p_\omega(t_0,v)}=0$.
  But this is clear from Assumption~1,
  \emph{i.e.}\ that all single site potentials are short range,
  and equation~\eqref{eq:bigdistance}, which together imply
  that the curvature along the axis vanishes.
\end{proof}

With this tools at hand it is easy to prove
\begin{proposition}\label{prop:periodicdense}
  The set of periodic orbits of $\Phi_\omega\colon\bR\times T_1\Mo\to T_1\Mo$
  is dense in $T_1\Mo$ for $\beta$--almost every $\omega\in\Omega$.
\end{proposition}
\begin{proof}
  To any nonempty open set $U\subseteq T_1\Mo$
  \nomenclature[AU]{$U$}{open subset of $T_1\Mo$}{}{}%
  we construct a periodic orbit whose trajectory intersects~$U$.
  \Cref{lemma:singsappear} guarantees that we hit
  a branch point at some time~$t>0$.
  We choose a branch point
  \begin{equation*}
    x\;\in\;\exp_{\Mo}(tU)
      \isect\pi_\omega^{-1}(\cS_\omega)\;\subseteq\;\Mo\text.
  \end{equation*}
  Now $\Phi_\omega$ is continuous, which shows that
  \begin{equation*}
    U':=-\Phi_\omega(t_1,U)\isect T_{1,x}\Mo
  \end{equation*}
  is open in $T_{1,x}\Mo$.
  Again by \cref{lemma:singsappear} we find $v\in U'$ and $t'>0$
  such that $\Phi_\omega(t',v)$ hits again a branch point.
  By construction the trajectory of~$v$ intersects~$U$ and is periodic.
\end{proof}

To gain more overview we introduce the universal covering
$\mathbf\pi_\omega^*\colon\Mo^*\to\Mo$ of $\Mo$
\nomenclature[Gpistaromega]{$\mathbf\pi_\omega^*$}{universal covering of~$\Mo$}{}{}%
\nomenclature[AMstaromega]{$\Mo^*$}{universal covering surface of~$\Mo$}{}{}%
and equipp it with the riemannian metric
$\mathbf g_{\omega,E}^*:=(\mathbf\pi_\omega^*)^*\mathbf g_{\omega,E}$.
\nomenclature[AgstaromegaE]{$\mathbf g_{\omega,E}^*$}{Jacobi metric on~$\Mo^*$}{}{}%
This makes $\Mo^*$~a Hadamard manifold.
We denote the natural length metric with $d_{\Mo^*}$.
\nomenclature[Adstaromega]{$d_{\Mo^*}$}{natural length on~$\Mo^*$}{}{}%
The fact that there is a certain amount of negative curvature in~$\Mo^*$
is expressed in the following.
\begin{proposition}\label{prop:visibility}
  For $\beta$--almost all $\omega\in\Omega$
  the riemannian surface $(\Mo^*,\mathbf g_\omega^*)$
  is a \emph{visibility manifold}, i.e.\ for all $p\in\Mo^*$
  and $\ve>0$ there exists $r=r(p,\ve)>0$ such that all geodesic segments
  \nomenclature[Gsigma]{$\sigma$}{geodesic segment}{}{}%
  $\sigma\colon[a,b]\to\Mo^*$ with distance at least~$r$ from~$p$
  are seen from~$p$ under an angle less then~$\ve$ (cf.\ \cite{EO73}):
  \begin{equation*}
    \forall p\in\Mo^*,\ve>0\exists r>0\colon\quad
    d_{\Mo^*}(p,\sigma[a,b])\ge r
    \quad\implies\quad
    \sphericalangle_p\bigl(\sigma(a),\sigma(b)\bigr)<\ve\text.
  \end{equation*}
  \nomenclature[Gangelp]{$\sphericalangle_p$}{angel in $p$}{}{}%
\end{proposition}
\begin{proof}
  We need to show that the set of all $\omega\in\Omega$
  which allow a positive~$\ve$ such that we find for all $r>0$
  a geodesic segment $\sigma_r\colon[a_r,b_r]\to\Mo^*$ with
  \begin{inparaenum}
    \item $d_{\Mo^*}(p,\sigma_r)>r$ and
    \item $\sphericalangle_p\bigl(\sigma_r(a_r),\sigma_r(b_r)\bigr)\ge\ve$
  \end{inparaenum}
  is of $\beta$--measure~$0$.
  To do this, we find a cone without singularities
  and invoke \cref{lemma:singsappear}.
  \par
  By compactness of $T_{1,p}\Mo^*$
  we get $x,y\in T_{1,p}\Mo^*$
  with $\sphericalangle_p(x,y)>\ve$ and
  \begin{equation*}
    d_{\Mo^*}(p,\gamma_t)\xto{t\to\infty}\infty\text,
  \end{equation*}
  where $\gamma_t\colon[a_t,b_t]\to\Mo^*$
  is the unique geodesic segment connecting
  $\gamma_t(a_t)=\exp_{\Mo^*}(tx)$ and
  $\gamma_t(b_t)=\exp_{\Mo^*}(ty)$.\\
  The theorem of Gauss--Bonnet implies that the curvature
  integrated over the cone
  \begin{equation*}
    \blacktriangle:=
      \Union_{t>0}\gamma_t[a_t,b_t]
  \end{equation*}
  between~$x$ and~$y$ is bounded by~$-\pi$ from below.
  This means that $\blacktriangle$ cannot cover two singularities,
  since then $\blacktriangle$, beeing geodesically convex,
  would cover the periodic orbit connecting these two,
  and this would contradict Gauss--Bonnet.
  Thereby $\blacktriangle$ contains a cone without singularity
  and \cref{lemma:singsappear} applies.
\end{proof}

\begin{remark}\label{rem:visibility}
  Two geodesic rays $\gamma,\gamma'\colon\bR_+\to\Mo^*$
  are called \emph{asymptotic}, if their distance is bounded.
  Being asymptotic is an equivalence relation,
  and the set of equivalence classes
  (equipped with a suitable topology, see \cite{EO73})
  is the \emph{ideal boundary} $\partial\Mo^*\equiv\Mo^*(\infty)$
  \nomenclature[AMhatomegastar]{$\partial\Mo^*\equiv\Mo^*(\infty)$}{ideal boundary of $\Mo^*$}{}{}%
  of $\Mo^*$.
  A point $\gamma(\infty):=[\gamma]\in\Mo^*(\infty)$
  \nomenclature[Ggammainfty]{$\gamma(\infty)$}{end point of~$\gamma$}{}{}%
  is a \emph{zero point}, if for all $\gamma'\in\gamma(\infty)$
  the distance between $\gamma$ and $\gamma'$ vanishes.
  \cite{EO73} explains that a visibility manifold
  satisfies \emph{Axiom~1}, i.e.\ any two distinct boundary points
  can be connected with a (not necessarily unique) geodesic.
\end{remark}

\begin{lemma}\label{lemma:zeropoint}
  For all $\omega\in\Omega$ holds the following.
  Every lift of a periodic geodesic in $\Mo$ to $\Mo^*$
  connects two zero points of $\Mo^*(\infty)$.
\end{lemma}
\begin{proof}
  Given a periodic geodesic $\gamma\colon S^1\to\Mo$,
  its lift $\gamma^*\colon\bR\to\Mo^*$, and another geodesic
  $\gamma'\colon\bR\to\Mo^*$ with $\gamma^*(\infty)=\gamma'(\infty)$,
  we integrate curvature over the area bounded by $\gamma^*(\bR_+)$,
  $\gamma'(\bR_+)$ and the geodesic segment connecting $\gamma^*(0)$
  and $\gamma'(0)$.
  Gauss--Bonnet assures us, that this quantity is bounded,
  and this is only possible if $\gamma^*$ and $\gamma'$ have distance~$0$.
  Otherwise, since $\gamma^*$ covers a periodic orbit, the curvature
  integral is unbounded.
\end{proof}


\begin{proof}[of \cref{thm:toptrans}]\label{proof-toptrans}
  \Cref{prop:periodicdense} tells us that
  for almost all $\omega\in\Omega$ periodic orbits of $\Phi_\omega$
  are dense in $T_1\Mo$.
  Therefore for any two open and nonempty sets
  $U,V\subseteq T_1\Mo$ we find
  two periodic geodesics $\gamma_U,\gamma_V\colon\bR\to\Mo$
  which intersect~$U$ and~$V$, respectively,
  i.e.\ $\dot\gamma_U(\bR)\isect U\ne\es$ and $\dot\gamma_U(\bR)\isect V\ne\es$.
  \par
  There are lifts $\gamma_U^*,\gamma_V^*\colon\bR\to\Mo^*$
  of $\gamma_U$ and $\gamma_V$ to the universal covering
  $\mathbf\pi_\omega^*\colon\Mo^*\to M_\omega$ of $\Mo$.
  By \cref{prop:visibility,rem:visibility}
  the endpoints $\gamma_U^*(-\infty)$ and $\gamma_V^*(\infty)$
  can be joined with a geodesic $\gamma^*\colon\bR\to\Mo^*$,
  i.e.\ $\gamma^*(-\infty)=\gamma_U^*(-\infty)$ and
  $\gamma^*(\infty)=\gamma_V^*(\infty)$.
  \Cref{lemma:zeropoint} makes sure that the distance between
  $\gamma^*(t)$ and $\gamma_U^*\bigl((0,\infty)\bigr)$
  respectively $\gamma_V^*\bigl((-\infty,0)\bigr)$ vanishes,
  as $t\to\pm\infty$, respectively:
  \begin{equation*}
    \lim_{t\to\infty}d_{\Mo^*}
      \bigl(\gamma^*(t),\gamma_U^*\bigl((0,\infty)\bigr)\bigr)
    =\lim_{t\to-\infty}d_{\Mo^*}
      \bigl(\gamma^*(t),\gamma_V^*\bigl((-\infty,0)\bigr)\bigr)
    =0\text.
  \end{equation*}
  This implies that $T\mathbf\pi_\omega^*\circ\dot\gamma^*(\bR)$
  intersects~$U$ and~$V$.
\end{proof}

%
\section{Ergodicity of the finite factor}\label{sect10}
%
Similar to the lattice and the poissonian case,
the Coulombic system does have a finite factor.
Assuming $(\vartheta,\beta)$ to be ergodic again
and the single site potentials to be $\eta\ge3$
times continuously differentiable, we
show ergodicity of this factor.
This is analogous to a result in \cite{Len03} in the setting of billiards.
Thanks to the smoothness of our system,
our proof can rely directly on \cite{Hop41}
without the need for technical generalisations
to cope with singularities like \emph{e.g.} in \cite{LW95}.

The following construction of the finite factor
is carried out in detail in \cite{Sch04,Sch10}.
As above the geodesic motion on the energy surface
is regularised with the twofold covering~\eqref{eq:Momega}.
There is one nontrivial deck transformation
\begin{equation*}
  \mathbf G_\omega\colon\mathbf M_\omega\to M_\omega\textq,
    G(q,Q):=(q,-Q)\text.
\end{equation*}
\nomenclature[AGomega]{$\mathbf G_\omega$}{deck transformation}{}{}%
For each $\omega\in\Omega$ and every $\ell\in\cL$
there exists an isometry
\begin{equation*}
  \mathbf\phi_{\omega,\ell}\colon
    \mathbf M_\omega\to\mathbf M_{\vartheta_\ell\omega}
\end{equation*}
\nomenclature[Gphiomegal]{$\mathbf\phi_{\omega,\ell}$}{isometry $\mathbf M_\omega\to\mathbf M_{\vartheta_\ell\omega}$}{}{}%
that shifts $q$ to $q-\ell$:
\begin{equation*}
  \pi_{\vartheta_\ell\omega}\circ\mathbf\phi_{\omega,\ell}(q,Q)
    \;=\;q-\ell\qquad\bigl((q,Q)\in\mathbf M_\omega\bigr)\text.
\end{equation*}
The deck transformation gives another such map
$\mathbf\phi_{\omega,\ell}\circ\mathbf G_\omega
  =\mathbf G_{\vartheta_\ell\omega}\circ\mathbf\phi_{\omega,\ell}$.
Before we can bundle the maps
$\bigl((\mathbf\phi_{\omega,\ell})_{\omega\in\Omega}\bigr)_{\ell\in\cL}$
into a group action of~$\cL$
on~$\mathbf M:=\Union_{\omega\in\Omega}\{\omega\}\times\mathbf M_\omega$,
\nomenclature[AM]{$\mathbf M$}{twofold branched covering of extended phase space}{}{}%
we have to divide through the group of deck transformations.

But $\pi_\omega$ is a branched covering,
so to keep the differentiable structure we first follow \cite{Kna87,KK92}
and restrict the geodesic flow to the energy surface
\begin{equation*}
  \mathbf\Sigma_{\omega,E}
    :=\big\{(x,v)\in T\bC^2\mid\mathbf g_{\omega,E}(v,v)=1\big\}\text.
\end{equation*}
\nomenclature[GSigmaomegaE]{$\mathbf\Sigma_{\omega,E}$}{unit sphere bundle of branched covering of $\mathbf M_\omega$}{}{}%
We divide $\mathbf\Sigma_{\omega,E}$ by the group~%
$\bZ_2\cong\big\{\Id_{\mathbf\Sigma_{\omega,E}},
  D\mathbf G_\omega\rstr_{\mathbf\Sigma_{\omega,E}}\big\}$
and get the smooth manifold
\begin{equation*}
  \Sigma_{\omega,E}:=\mathbf\Sigma_{\omega,E}/\bZ_2\text.
\end{equation*}
\nomenclature[GSigmaomegaE]{$\Sigma_{\omega,E}$}{$\bZ_2$--quotient of $\mathbf\Sigma_{\omega,E}$}{}{}%
The disjoint union
\begin{equation*}
  \mathbf\Sigma_E
    :=\Union_{\omega\in\Omega}\{\omega\}\times\mathbf\Sigma_{\omega,E}
\end{equation*}
\nomenclature[GSigmaE]{$\mathbf\Sigma_E$}{$\bZ_2$--quotient of $\mathbf\Sigma_E$}{}{}%
inherits its topology from the embedding into $\Omega\times T\bC^2$.
The quotient
\begin{equation*}
  \Sigma_E:=\mathbf\Sigma_E/\bZ_2
    =\Union_{\omega\in\Omega}\{\omega\}\times\Sigma_{\omega,E}
\end{equation*}
is the phase space for which we can construct the desired group action
\begin{equation*}
  \Theta^E\colon\cL\times\Sigma_E\to\Sigma_E\textq,
  \Theta^E_\ell\bigl(\omega,[x]_{\bZ_2}\bigr)
    :=\bigl(\vartheta_\ell\omega,[D\mathbf\phi_{\omega,\ell}(x)]_{\bZ_2}\bigr)\text.
\end{equation*}
\nomenclature[GThetaE]{$\Theta^E$}{$\cL$--action on $\Sigma_E$}{}{}%
This group action is continuous.
Analogous to \cref{sec2} the quotient
\begin{equation*}
  \pi_\cL\colon\Sigma_E\to\hat\Sigma_E:=\Sigma_E/\cL
\end{equation*}
\nomenclature[GpiL]{$\pi_\cL$}{quotient $\Sigma_E\to\widetilde\Sigma_E$}{}{}%
\nomenclature[GSigmahatE]{$\hat\Sigma_E$}{compact coulomb phase space}{}{}%
is compact and metrizable, inherits the geodesic flow
\begin{equation*}
  \hat\Phi_E\colon\bR\times\hat\Sigma_E\to\hat\Sigma_E
\end{equation*}
\nomenclature[GPhihatE]{$\hat\Phi_E$}{geodesic flow on $\hat\Sigma_E$}{}{}%
and carries the finite and $\hat\Phi_E$--invariant
Liouville measure~$\hat\mu_E$, see \cite{Sch10}.
\nomenclature[GMstar]{$\hat\mu_E$}{Liouville measure on $\hat\Sigma_E$}{}{}%
An overview over the different phase spaces including
the embedding into $\Omega\times T\bC^2$
and the unit tangent bundle~$\mathbf\Sigma_{\omega,E}^*$
\nomenclature[GSigmastaromegaE]{$\mathbf\Sigma_{\omega,E}^*$}{unit tangent bundle of $\Mo^*$}{}{}%
of the universal cover~$\mathbf\pi_\omega^*\colon\Mo^*\to\Mo$
is given in \cref{abb:orientierung}.
\nomenclature[GPistaromega]{$\mathbf\Pi_\omega^*$}{covering $\mathbf\Sigma_{\omega,E}^*\to\mathbf\Sigma_{\omega,E}$}{}{}%

\begin{figure}[h]
    \begin{equation*}
    \begin{CD}
      \mathbf\Sigma_{\omega,E}^*@>\mathbf\Pi_\omega^*>>
		\mathbf\Sigma_{\omega,E}@>>>
                \mathbf\Sigma_E@>>>\Omega\times T\bC^2\\
      &&@VV\pi_{\omega,\bZ_2}V@VV\pi_{\bZ_2}V\\
      &&\Sigma_{\omega,E}@>\iota_{\omega,E}>>\Sigma_E
      @>\pi_{\cL}>>\hat\Sigma_E\phantom{\times T\bC^2}
    \end{CD}
  \end{equation*}
  \caption{Overview over the relations between the different phase spaces}
  \label{abb:orientierung}
\end{figure}

\begin{theorem}\label{thm:erg:Coul}
  Let the shift action $\vartheta\colon\cL\times\Omega\to\Omega$
  be ergodic with respect to~$\beta$,
  the potential $\tilde V_\omega$ three times continuous differentiable and
  $(\tilde V_\omega,E)$ of strictly negative curvature for
  $\beta$-almost all $\omega\in\Omega$.\\
  Then the flow~$\hat\Phi_E$ is ergodic with respect to~$\hat\mu_E$.
\end{theorem}

\begin{proof}
  In \cite{Hop41} the ergodicity of a recurrent and hyperbolic geodesic flow
  on a $2$-dimensional riemannian surface is proven.  We use that work here.

  We need to show that for every continuous function
  $\hat f\colon\hat\Sigma_E\to\bR$ the limits
  \begin{equation}\label{fbar}
    \bar f^\pm:=\lim_{T\to\pm\infty}\frac1T
      \int_0^T\hat f\circ\hat\Phi_E(t,\cdot)\,dt
  \end{equation}
  exist and are almost everywhere constant.
  The existence a.e.\ of~$\bar f^\pm$
  is guaranteed by Birkhoff's ergodic theorem,
  and so is $\bar f^+=\bar f^-=:\bar f$ $\hat\mu_E$-a.e..

  Via the quotients and embeddings, see \cref{abb:orientierung}, we introduce
  \begin{equation*}
    f_\omega^*\colon\mathbf\Sigma_{\omega,E}^*\to\bR\textq,
    f_\omega^*:=\hat f
      \circ\pi_\cL
      \circ\iota_{\omega,E}
      \circ\pi_{\omega,\bZ_2}
      \circ\mathbf\Pi_\omega^*
      \quad(\omega\in\Omega)\text.
  \end{equation*}
  By Poincar\'e's recurrence theorem
  the finite measure preserving dynamical system
  $(\hat\Sigma_E,\hat\Phi_E,\hat\mu_E)$ is recurrent, so that
  $\liminf_{t\to\infty}\abs{f_\omega^*\circ(\mathbf\Phi_{\omega,E}^*)_t
    -f_\omega^*}=0$ a.e.,
    see \emph{e.g.} \cite[p.~14--16]{Aar97}.
  Hyperbolicity is guaranteed by the curvature assumption.
  Using the a.e.\ constancy of $\bar f_\omega^*$
  along stable and unstable manifolds like in \cite{Hop41}
  we see, that for each $\omega\in\Omega$ the limit~$\bar f_\omega^*$
  is almost surely constant.

  Finally we use the ergodicity of $(\vartheta,\beta)$
  to conclude that $\bar f$ ist constant on~$\hat\Sigma_E$.
  By denoting this a.e.-value by $F(\omega)$
  for $\beta$-a.a.\ $\omega\in\Omega$
  we define $\beta$-a.e.\ a function $F\colon\Omega\to\bR$.
  Since the whole construction depends measurably on $\omega\in\Omega$,
  $F$~itself is measurable.
  To see that $F$ is in fact $\vartheta$--invariant, we introduce
  \begin{equation}\label{fhat}
    f_{\Sigma_E}:=\bar f\circ\pi_\cL
      \colon\Sigma_E\to\bR
    \qquad(\omega\in\Omega)
  \end{equation}
  and write for a.e.\ $\omega\in\Omega$,
  a.e.\ $x\in\Sigma_{\vartheta_\ell\omega}$ and a.e.\ $y\in\Sigma_\omega$
  \begin{align*}
    F(\vartheta_\ell\omega)&
      =f_{\Sigma_E}\circ\iota_{\vartheta_\ell\omega}(x)
      =f_{\Sigma_E}\circ\Theta^E\circ\iota_\omega(y)
      =f_{\Sigma_E}\circ\iota_\omega(y)
      =F(\omega)\text,
  \end{align*}
  using that $f_{\Sigma_E}$~by definition is $\Theta^E$--invariant,
  see \cref{fbar,fhat}.
  Therefore, by ergodicity of $(\vartheta,\beta)$,
  $F$~is $\beta$-a.s.\ constant,
  which then implies that $\bar f$~is constant modulo~$\hat\mu_E$.
\end{proof}
The above theorem relates to a geodesic flow on 
a (extended and compactified) unit tangent bundle.
This construction has been made in order to make use of the
well-known consequences of negative curvature on the dynamics.\\
Like in the case of smooth potentials treated in \cref{sec5},
we are finally interested in the (analog of) 
the Hamiltonian flow \eqref{tPhiE} generated by $\hat{H}$ (based on the Hamiltonian $H\colon P\to \bR$ on the extended phase space~$P$
of the first regularization in \cref{sec9}).
As this flow is related to the geodesic flow of \cref{thm:erg:Coul}
by a continuous time change, we get:
\begin{corollary}\label{cor:erg:Coul}
Under the conditions of \cref{thm:erg:Coul} 
the coulombic Hamiltonian flow \eqref{tPhiE}
on the compactified energy surface $\hS_E$
is ergodic w.r.t.\ Liouville measure $\tlambda_E$.
\end{corollary}
\begin{remark}[Markov Partition]\label{todo}
In a forthcoming article, one of us (CS) shows 
for all large energies, under some additional conditions,
the existence of Markov partitions for the system, cf.~\cite{Sch10}.
\end{remark}

In \cite{CLS10} Cristadoro, Lenci and Seri show recurrence for 
random Lorenz tubes. Their configuration spaces are contained in a connected 
union of translates of a fundamental polygon. 
We a have similar statement for the Coulomb case. 

Again with $H\colon P\to \bR$ from \cref{sec9},
for $\ell\in\cL\setminus \{0\}$ we define the factors 
\[\check\Sigma_E:=H^{-1}(E)/\mathrm{span}_{\bZ}(\ell)
\qtextq{and}
\check\Sigma_{E,\omega}:=H_\omega^{-1}(E)/\mathrm{span}_{\bZ}(\ell)\text,\] 
for $\ell$--periodic $\omega\in\Omega$.
\nomenclature[AXhat]{$\check\Sigma_{E,\omega}$}{AXhat...}{}{}%
\nomenclature[AXhatomega]{$\check\Sigma_{E,\omega}$}{Coulomb tube}{}{}%
Again we assume $(V_\omega,E)$ to be of strictly negative curvature $\beta$--a.s..
As we divided by the free $\Theta$--action of a subgroup of~$\cL$,  
by the analog of \eqref{B} we obtain smooth flows on the non--compact
three--manifolds~$\check\Sigma_{E,\omega}$.

We assume for simplicity that $\ell\in\cL$ is primitive.
Then $\cL/\mathrm{span}_{\bZ}(\ell)\cong \bZ$ 
(and via the element $\ell'$ of a positively oriented basis 
$(\ell,\ell')$ of $\cL$, this isomorphism is unique).
We may consider random potentials indexed by $\Omega_\ell:=J^\bZ$.
\nomenclature[GOmegaell]{$\Omega_\ell$}{single site potential configurations on $\check\Sigma_E$}{}{}%
$\bZ$--ergodic probability measures $\beta_\ell$
\nomenclature[Gbetal]{$\beta_\ell$}{probability measure on $\Omega_\ell$}{}{}%
on $\Omega_\ell$ give rise to $\cL$--ergodic image measures~$\beta$
on $\Omega$, via the injection $\mathbf{i}\colon\Omega_\ell\to\Omega$,
${\bf i}(\omega)_{a\ell+b\ell'}=\omega_b$.
\nomenclature[Ai]{$\mathbf i$}{injection $\Omega_\ell\to\Omega$}{}{}%
These are analogues of the measures considered in \cite{CLS10}.
For non--trivial $\beta_\ell$ the flow on~$\check\Sigma_E$ is not ergodic
w.r.t.\ the measure induced by~$\check\mu_E$
via the projection $\check\Sigma_E\to \hat\Sigma_E$.

\begin{proposition}[Ergodicity and Recurrence of Coulomb Tubes] 
Under\\ the conditions of \cref{thm:erg:Coul} 
\begin{compactitem}
  \item $\beta$--a.s.\ the motion on~$\check\Sigma_E$ is recurrent.
  \item For ${\bf i}(\beta_\ell)$--almost all $\omega\in \Omega$
	the motion on~$\check\Sigma_{E,\omega}$ is recurrent and ergodic.
\end{compactitem}
\end{proposition}

\begin{proof}
By \cref{cor:erg:Coul} the flow~$\hat\Phi_E$
is ergodic w.r.t.\ $\tlambda_E$. Thus by the analog of 
\cref{prop:ConsequencesOfErgodicity} for Coulomb potentials,  
for $\beta$--a.e.~$\omega$
the asymptotic velocity $\ovv(x)=0$ a.e.\ on~$\Sigma_{E,\omega}$, 
and thus on $\check\Sigma_E$.
If $\beta$ is of the form ${\bf i}(\beta_\ell)$, then for $\beta$--a.e.~$\omega$
the asymptotic velocity $\ovv(x)=0$ a.e.\ on~$\check\Sigma_{E,\omega}$.

We compare with Schmidt \cite{Sch98} (and the references cited therein).
In the case of an ergodic transformation~$T$ on 
a standard probability space $(X,\mu)$, a Borel map $f\colon X\ar\bR$
and the induced cocycle (orbit sum) $\hat f\colon\bZ\times X\ar\bR$, he defines 
$\sigma_k(A):=\mu(\{x\in X\mid\hat f(k,x)/k\in A\})\quad(k\in\bN)$,
for Borel sets $A\subseteq \bR$.

Under the assumption that vaguely $\lim_{k\ar\infty}\sigma_k=\delta_0$,
he notes that $f$~is \emph{recurrent}, meaning that 
$\liminf_{k\ar\infty} \abs{\hat f(k,x)}=0$, $\mu$--a.s..

Here we use a Poincar\'e map discretization $(X,\mu,T)$
of the flow $\hat\Phi_E$ on~$\hat\Sigma_E$
(by \cref{todo} such a discretization exists and is ergodic).
For $f(x_0)$, with $x_0=[(\omega, p_0,q_0)]\in X$, we use the difference 
$\LA q(T(x_0))-q_0,\ell^\perp\RA$ in position along the direction 
$\ell^\perp:=\bsm 0&1\\-1&0\esm\ell\in \bR^2$ perpendicular to $\ell$.
$f$ is well-defined on~$X$. 

The assumptions of \cite{Sch98} apply and $f$~is recurrent. 
This means that for a.e.\ $x_0$ 
the motion in the Coulomb tube returns infinitely often to every neighbourhood
of the circle in the configuration torus given by $\LA q_0,\ell^\perp\RA=0$.
Then one obtains, by Poincar\'e's Recurrence Theorem for the induced map, 
recurrence in the usual sense.

Ergodicity of the flow on~$\check\Sigma_{E,\omega}$ follows from local ergodicity,
using Hopf's argument \cite{Hop39}.
\end{proof}

%
%
\printnomenclature
\addcontentsline{toc}{section}{References}
\bibliographystyle{alpha}
\bibliography{motion}

\end{document}